\documentclass[10pt]{article}

\usepackage{authblk}
\usepackage[utf8]{inputenc}    
\usepackage[T1]{fontenc}

\usepackage{graphicx}
\usepackage{amsthm,amsmath}
\usepackage{amssymb}
\usepackage{amsfonts}
\newtheorem{theorem}{Theorem}
\newtheorem{proposition}{Proposition}
\newtheorem{lemma}[theorem]{Lemma}

\usepackage{tikz}
\usepackage{comment}
\usepackage[absolute,overlay]{textpos}         
\DeclareGraphicsExtensions{.png,.pdf,.eps}

\makeatletter

\newcommand{\Rmnum}[1]{\expandafter\@slowromancap\romannumeral #1@}
\makeatother

\providecommand{\keywords}[1]{\textbf{\textit{Keywords:}}#1}

\newcommand\blfootnote[1]{%
  \begingroup
  \renewcommand\thefootnote{}\footnote{#1}%
  \addtocounter{footnote}{-1}%
  \endgroup
}

\title{Mixing Time on the Kagome Lattice\footnote{This work was supported by the ANR project QuasiCool (ANR-12-JS02-011-01)}}
\author{Alexandra Ugolnikova\footnote{contact: alexandra.ugolnikova \textit{at} lipn.univ-paris13.fr}}

\affil{Laboratoire d'Informatique de Paris Nord}
\date{\today}

\begin{document}

\maketitle

\begin{abstract}
We consider tilings of a closed region  of the Kagome lattice (partition of the plane into regular hexagons and equilateral triangles such that each edge is shared by one triangle and one hexagon). We are interested in the rate of convergence to the  stationarity of a natural Markov chain defined on the set of Kagome tilings. The rate of convergence can be represented by the mixing time which mesures the amount of time it takes the chain to be close to its stationary distribution. We obtain a $O(N^4)$ upper bound on the mixing time of a weighted version of the natural Markov chain. We also consider Kagome tilings restrained to two prototiles, prove flip-connectivity and draw a $O(N^4)$ upper bound as well on the mixing time of the natural Markov chain in a general (non weighted) case. Finally, we present simulations that suggest existence of a long range phenomenon. 
\end{abstract}


\blfootnote{\textup{2010} \textit{Mathematics Subject Classification}: 60J10, 52C20.}

\section{Introduction}

In the present work we consider tilings on the \textbf{Kagome} lattice (also known as the \textbf{Butterfly}/\textbf{tri-hexagonal} lattice) which is a partition of the plane into hexagons and triangles of the same side length such that each edge is shared by one triangle and one hexagon (see Figure \ref{fig: lattice}). The dimer model on the tri-hexagonal lattice was studied along with other dimer models on regular lattices \cite{WW07} (dimers on square and triangular lattice were studied, \textbf{e.g.}, in \cite{EKLP92, E84, LRS95}).

A prototile on the Kagome lattice is a hexagon with two adjacent triangles. There are three type of prototiles (up to rotation) that we call \textbf{trapeze, fish} and \textbf{lozenge} and that are shown in Figure \ref{fig: all tiles}. Note that the defined prototiles are different from dimers and correspond more to trimers. A closed region $R$ of the Kagome lattice is tileable if there is at least one partition of the cells of this region into prototiles. An example of a Kagome tiling is shown in Figure  \ref{fig: tiling}.

\begin{figure}[hbtp] 
\centering
\includegraphics[width=0.5\textwidth]{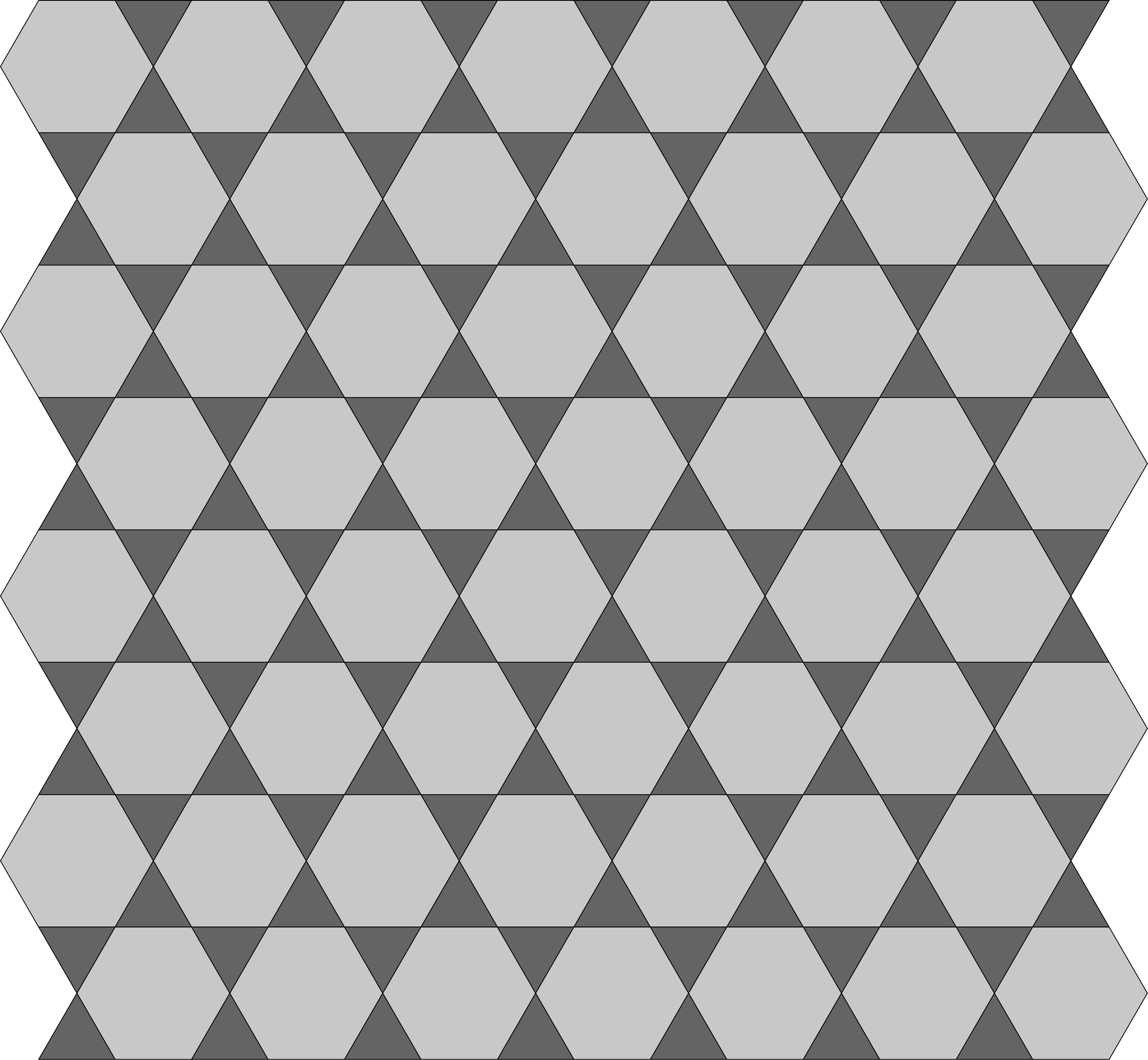}
\caption{Kagome lattice}
\label{fig: lattice}
\end{figure}

\begin{figure}[hbtp] 
\centering
\includegraphics[width=0.5\textwidth]{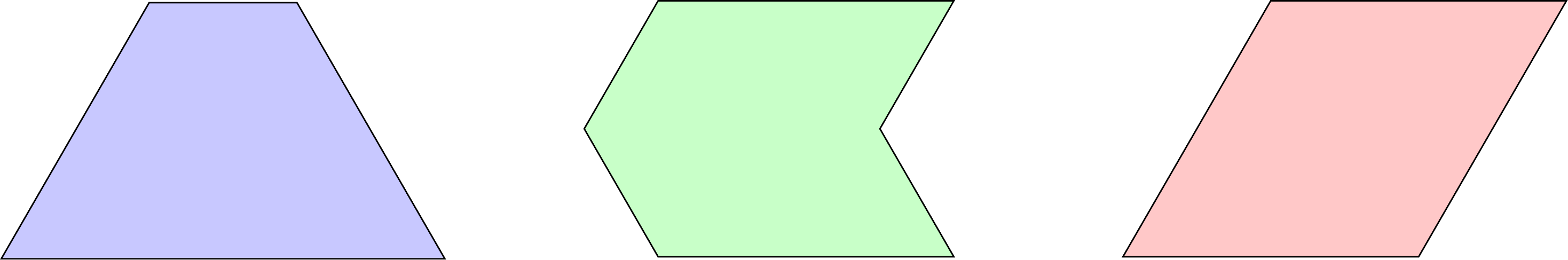}
\caption{Three Kagome prototiles: trapeze, fish and lozenge.}
\label{fig: all tiles}
\end{figure} 

\begin{figure}[hbtp] 
\centering
\includegraphics[width=0.5\textwidth]{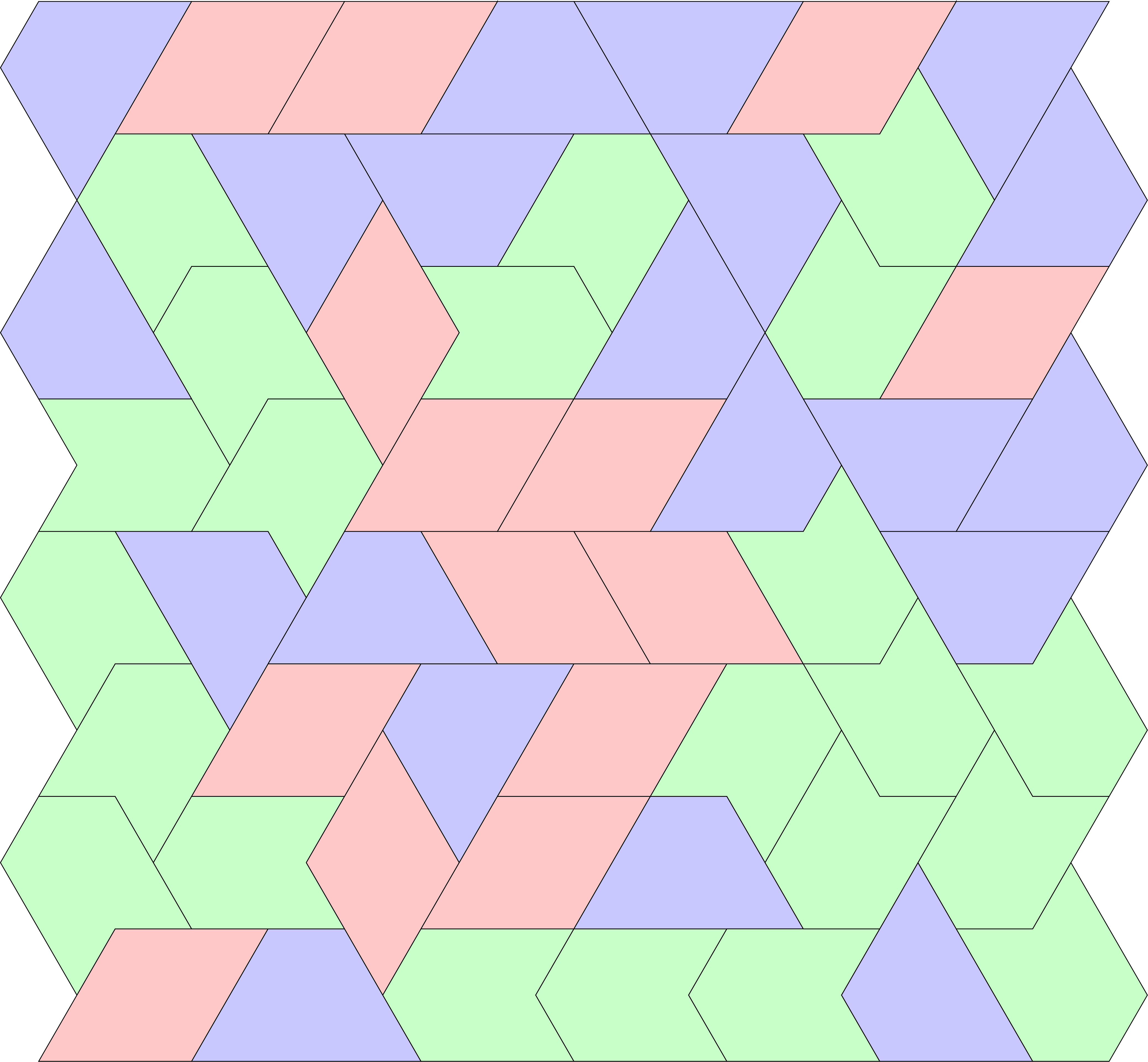}
\caption{General Kagome tiling}
\label{fig: tiling}
\end{figure} 

We are trying to answer questions related to random generation. In Section \ref{sec: settings} we lay down the main notions about Markov chains and mixing time. We introduce the height function and flips as it was done in \cite{B06}. The connectivity of the configurational space under flips (that was proven in \cite{B06} as well) allows us to introduce a Markov chain and study its properties in Section \ref{sec: markov chain}. Even though simulations using Coupling from the Past (see \cite{LPW09}, Chapter 22) suggest the mixing time to be $O(N^{2.5})$, where $N$ is the number of tiles (see Table \ref{tab: simulation})), a rigorous proof seems to be rather challenging and we explain why it is so at the end of Section \ref{sec: markov chain}.

In Section \ref{sec: weighted chain} we consider a weighted version of the chain  where weights $\lambda$ are assigned to one type of prototiles. Adding weights helps to notice phase transitions in the system. Usually weights $\lambda$ are assigned in such a way that the case $\lambda = 1$ corresponds to the unweighted version of the chain. This case is generally hard to analyze but it becomes easier to analyze the dynamics for $\lambda$ below and above some critical point. We draw polynomial bounds on the mixing time of the Markov chain  with a condition on weights $\lambda$ in Theorem \ref{thm: weighted kagome mixing}. 

In Section \ref{sec: restrained} we consider tilings by only two specific types of tiles instead of three and consider a version of the initial chain. Fast mixing can be proven when the state space (the set of all possible configurations) is connected. We show that in case of lozenge-type regions, the state space is connected. We prove that the Markov chain introduced for restrained Kagome tilings is fast mixing for lozenge-shaped regions (Theorem \ref{th : fast mixing restr}). 

Finally, in Section \ref{sec: limit shape} we present some simulations for a lozenge-type region that suggest a specific limit shape similar to the Arctic circle for dimer tilings.

\begin{table}
\begin{center}
\begin{tabular}{|r|c|}
\hline 
$n$ & $t$ \\
\hline
10 & $ 2.91 \times 5^5$ \\
12 & $ 2.53 \times 6^5$ \\
14 & $ 1.92 \times 7^5$ \\
16 & $ 1.91 \times 8^5$ \\
18 & $ 1.76 \times 9^5$ \\
20 & $ 1.43 \times 10^5$ \\
22 & $ 1.44 \times 11^5$ \\
24 & $ 1.35 \times 12^5$ \\  
26 & $ 1.31 \times 13^5$ \\
28 & $ 1.15 \times 14^5$ \\
30 & $ 1.07 \times 15^5$ \\
32 & $ 1.13 \times 16^5$ \\
34 & $ 1.05 \times 17^5$ \\
36 & $ 1.05 \times 18^5$ \\
38 & $ 1.04 \times 19^5$ \\
40 & $ 1.02 \times 20^5$ \\
\hline 
\end{tabular}
\end{center}
\caption{Coupling time for Kagome tilings of the square region of size $2n$, $n = 5, \ldots 20$ with $N = n^2$ tiles.}
\label{tab: simulation}
\end{table}

\section{Settings}
\label{sec: settings}

\subsection{Mixing and coupling times}
\label{subsec: mixing and coupling}

We consider reversible ergodic Markov chains with a finite state space $\Omega$. We denote its stationary distribution by $\pi$, its probability law by $P$. For any initial state $x \in \Omega$ let the \textbf{total variation distance} between $P(x, \cdot)$ and $\pi$ is 

$$
d_{TV}(P(x, \cdot), \pi) := \frac{1}{2} \sum_{y \in \Omega} \lvert P^t(x,y) - \pi(y) \rvert.
$$

Let us write it as $d_x(t)$. The \textbf{mixing time} of the $MC$ is the time it takes the chain to get close to its stationary distribution. Formally, it is defined as follows:

$$
\tau_{mix}(\varepsilon) := \max_{x \in \Omega} \min \lbrace t: d_x(t') \leq \varepsilon \;\forall\; t' \geq t\rbrace , 
$$

$$
\tau_{mix}:= \tau_{mix} \left( \frac{1}{4} \right).
$$

A classical way to bound the rate of convergence of a chain is to bound its mixing time. There are lot of different ways of bounding the mixing time: via the second largest eigenvalue (which can be analyzed using the corresponding tiling graph's properties), coupling methods (see, \textbf{e.g.}, \cite{DG98, DS91, LPW09, R06,S92}). 

Here we concentrate on the \textbf{coupling} method.
A coupling for two probability distributions $\mu$ and $\nu$ is a pair of random variables $(X,Y)$ defined on the  same probability space such that $\mathbb{P}(X=x) = \mu(x)$ and $\mathbb{P}(Y=y) = \nu(y)$. Here we  will be  using couplings for Markov chains  where constructing copies of the chain proves to be a useful tool to analyze the distance to stationarity. A \textbf{coupling of a MC} is a stochastic process $(X_t, Y_t)_t$ on $\Omega \times \Omega$ such that:
\begin{enumerate}
\item $X_t$ and $Y_t$ are copies of the MC with initial states $X_0 = x$ and $Y_0 = y$;
\item If $X_t = Y_t$, then $X_{t+1}  = Y_{t+1}$.
\end{enumerate}
Let $T^{x,y} = \min \lbrace t: X_t = Y_t \vert X_0 = x, Y_0 = y\rbrace$. Then define the \textbf{coupling time} of the $MC$ to be
$$
\tau_{cp} :=  \max_{x,y} \mathbb{E} T^{x,y}.
$$

The following result \cite{A81} relates the coupling and mixing times:

\begin{theorem}[Aldous]\label{theorem: A}
$$
\tau_{mix}(\varepsilon) \leq \lceil \tau_{cp} e\ln \varepsilon^{-1}  \rceil.
$$
\end{theorem}

One of the most used methods to bound the mixing time is the following path coupling theorem \cite{DG98}. The authors show that in order to bound the coupling time, one only has to consider pairs of configurations of the coupled chain that are close to each other in the defined metric. It is sufficient to prove that they [each pair of configurations] have more tendency to remain close to each other under the evaluation of the chain. Then the mixing time is polynomial and depends on the diameter of the corresponding graph. 

\begin{theorem}[Dyer-Greenhill]\label{theorem: DG}
Let $\varphi : \Omega \times\Omega \rightarrow \lbrace 0, \ldots, D \rbrace$  be an integer-valued metric,  $U$ -- a subset of $\Omega \times \Omega$ such that for all $(X_t, Y_t) \in  \Omega \times \Omega $ there exists a path between them: $X_t  = Z_0, Z_1, \ldots, Z_n = Y_t $  with $(Z_i, Z_{i+1}) \in U$ for $ 0\leq i\leq r-1 $ and $$ \sum_{i=0}^{r-1} \varphi(Z_i, Z_{i+1}) = \varphi(X_t, Y_t).$$
Let $MC$ be a Markov chain on $\Omega$ with transition matrix $P$. Consider a random function $f: \Omega \rightarrow \Omega$ such  that $\mathbb{P}(f(X)= Y) = P(X,Y) $ for all $X,Y \in \Omega$, and let a coupling be  defined by $(X_t,Y_t) \rightarrow (X_{t+1}, Y_{t+1}) = (f(X_t), f(Y_t))$.

\begin{enumerate}
\item If there exists $\beta < 1$  such that $ \mathbb{E}[\varphi(X_{t+1}, Y_{t+1})] \leq \beta \varphi(X_t,Y_t)$ for all $(X_t, Y_t) \in U$, then the mixing time satisfies $$ \tau_{mix}(\varepsilon) \leq \frac{\ln(D\varepsilon^{-1})}{1-\beta}.$$ 
\item If $\beta = 1$, and  $ \exists \; \alpha > 0$ that satisfies $\mathbb{P}(\varphi(X_{t+1},Y_{t+1}) \neq  \varphi(X_t,  Y_t)) \geq \alpha $ for all $t$ such that $X_t \neq Y_t$. Then the mixing time  satisfies  
$$ \tau_{mix}(\varepsilon) \leq  \left\lceil\frac{eD^2}{\alpha} \right\rceil \lceil \ln \varepsilon^{-1} \rceil .$$ 
\end{enumerate}
\end{theorem}

\subsection{Height function and flips}
\label{subsec: height function and flips}
Consider the Kagome lattice (Figure \ref{fig: lattice}). Following Thurston \cite{T90}, the orientation on the edges of the lattice is set to be clockwise on triangles and anti-clockwise on hexagons. We assign $+1$ to an edge if it belongs to a tile and $-2$ otherwise (\textbf{flow} in \cite{B06}). If we follow the edges  around a tile and sum the flows: with $+$ if we follow the orientation of arrows, with $-$ if not, then the flow gives $0$ around each tile. Let us now introduce the notion of \textbf{height}. We denote the set of lattice vertices in $R$ by $V_R$ and $ N := \vert V_R \vert$. Fix a vertex $v \in V_R $ and let its height $h(v):=0$. For any $w \in V_R$ its height $h(w)$ is defined by the flow from $v$ to $w$. See Figure \ref{fig: flow}. A vertex $v \in V_R$ is called a \textbf{local minimum (maximum)} if $h(v)$ is less(greater) than $h(w)$ for all $w$ that share an edge of the lattice with $v$.
For a tiling $T_R$ (or simply $T$) let its height be
$$
h(T) := \sum_{v \in V_R} h(v).
$$
Take two tilings $T$ and $Q$ of $R$. We define a natural order on the set of tilings as follows: We say that $T \leq Q$ if $h(T) \leq h(Q)$.

\begin{figure}[hbtp] 
\centering
\includegraphics[width=0.5\textwidth]{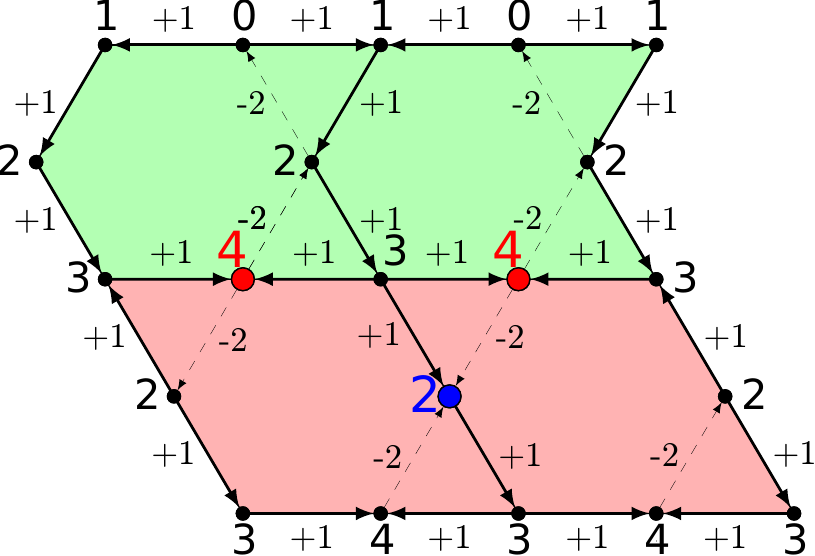}
\caption{Flow, orientation and heights of a Kagome tiling.}
\label{fig: flow}
\end{figure}

Let us have two tiles $a$ and $b$. Let $a$ consist of a hexagon $h_a$ and two triangles $t_a$ and $t_1$, $b$ of $h_b, t_b$ and $t_2$, where $t_a$ is adjacent to $h_b$, $t_b$ is adjacent to $h_a$. Then the transformation of $a$ and $b$ into $a'$, $b'$ such that $t_a \in b'$, $t_b \in a'$ is called a \textbf{(simple) flip} (as defined in \cite{B06}).
Flips can be performed only around local minima and maxima. A flip turns a local minimum into a local maximum and the other way around. We say that it has two directions: \textbf{minimal} if it decreases the height in the vertex and and \textbf{maximal} if it increases. All possible flips (up to rotation) are shown in Figure \ref{fig: all_flips}. An example of flips is shown in Figure \ref{fig: flip_path}.

\begin{figure}[hbtp] 
\centering
\includegraphics[width=0.8\textwidth]{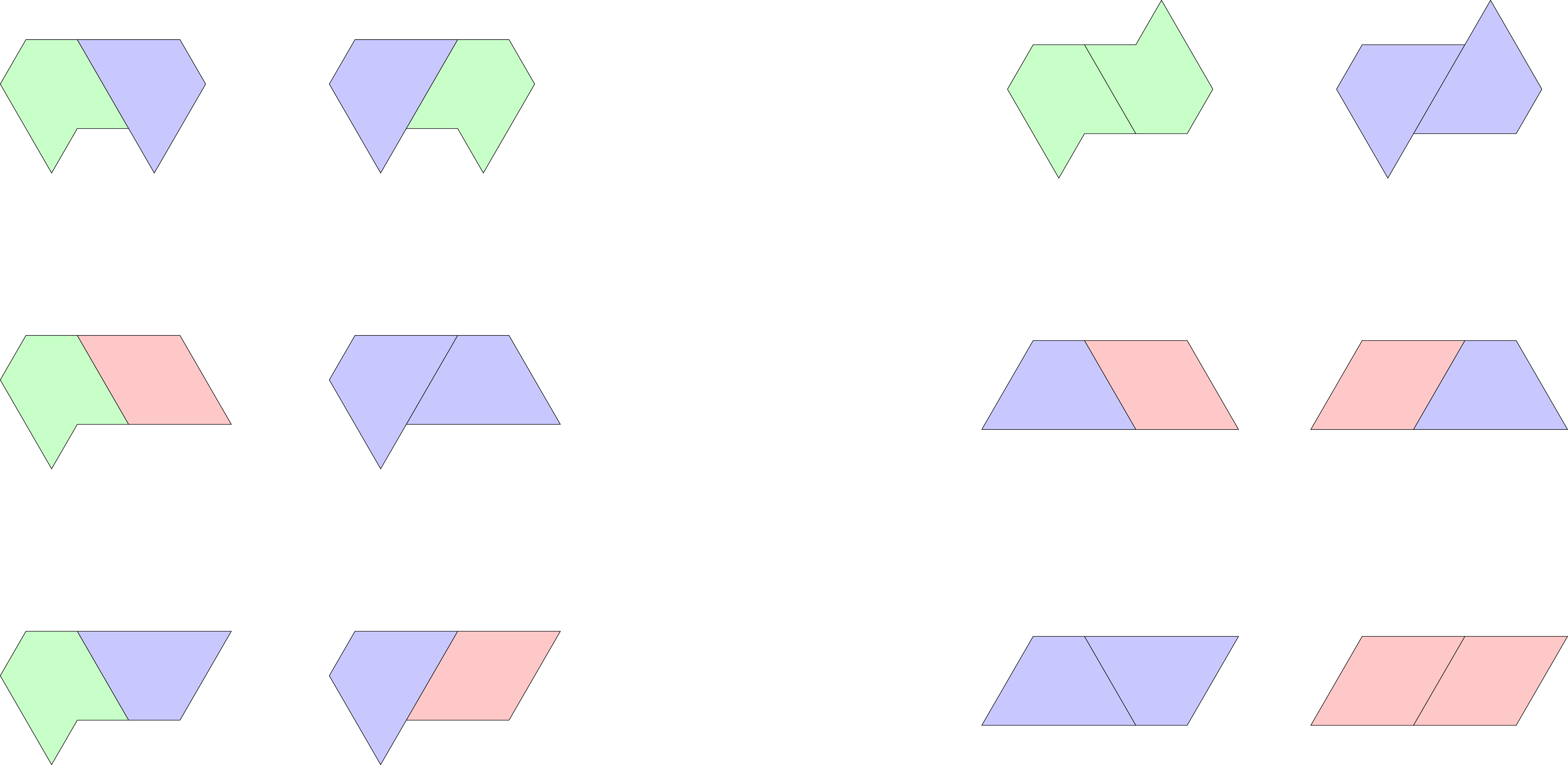}
\caption{All possible flips (up to rotation) for the general Kagome tiling.}
\label{fig: all_flips}
\end{figure} 

\begin{figure}[hbtp] 
\centering
\includegraphics[width=1\textwidth]{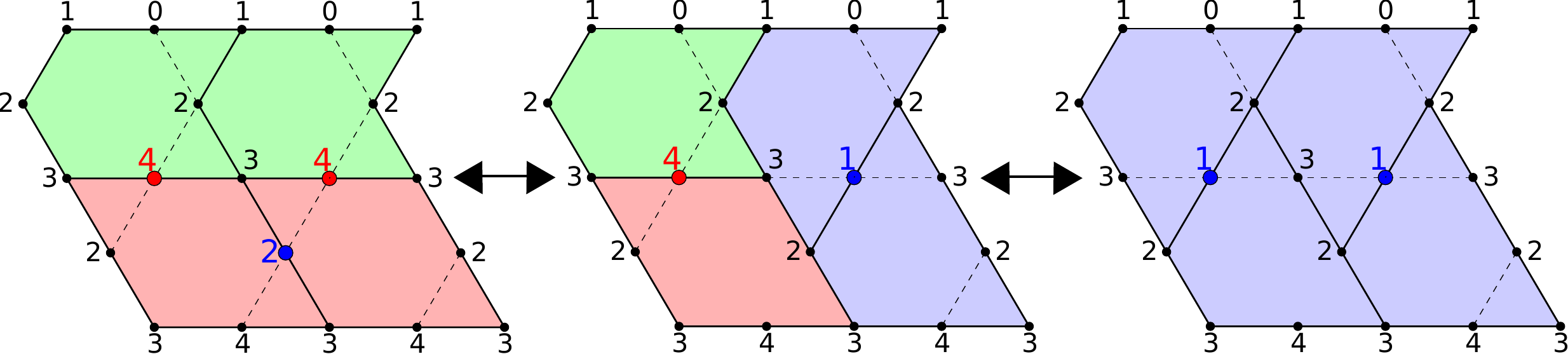}
\caption{Example of flips performed in a Kagome tiling.}
\label{fig: flip_path}
\end{figure}

\section{General Kagome Markov chain}
\label{sec: markov chain}

Let $\Omega_R$( or simply $\Omega$) be the set of Kagome tilings of the region $R$ with $N$ inner vertices and let $h$ be a height function defined on the vertices of $R$.

\begin{proposition}\label{prop: weight for vertices}
Let  $n_h(v)$ be the number of different heights $h$ in a vertex $v \in V_R$. Then for all $v \in V_{R}$ 
$$
n_h(v) \leq \sqrt{N},
$$
where $N = \vert V_R \vert.$ 
\end{proposition}
{\noindent \emph{Proof.}
In order to prove this, let us introduce the following subsets of vertices of our system: we denote $V_{\partial R}$ by $V_0$  which contains all vertices that belong to the border $\partial R$.  The border does not change, so $n_h(v) = 1$ for $v \in V_0$. Now take vertices in $V \setminus V_0$ that have vertices from $V_0$ as their neighbours, and denote this subset by $V_1$. Now  $n_h(v) \leq 2 $ for $v \in V_1$. Take the subset of vertices from $V \setminus (V_0 \cup V_1)$ that have neighbouring vertices from $V_1$ and denote it by $V_2$, etc. $R$ has a finite number of vertices,  so in the end we will obtain the subset $V_l$. Denote the perimeter of $R$ by $L(R)$.  Since $R$ has area $N$, then $l \leq \frac{N}{L(R)}$ and $\frac{N}{L(R)} \leq \sqrt{N}$. For all $(v,w) \in (V_i, V_{i+1}),\;  i = 0, ..., l-1: \,\lvert n_h(v)- n_h(w) \rvert \leq 1$, so  $ n_h(v) \leq \sqrt{N}.$

}{\hfill \hbox{\rlap{$\sqcap$}$\sqcup$}

Let $G$ be the tiling graph in which each vertex corresponds to a tiling of $R$, and two vertices are connected by an edge i.f.f. the corresponding tilings differ by one flip. If the graph has one connected component, there is the unique minimal tiling $T_{min}$ (having the minimal height) and the unique maximal tiling $T_{max}$ (having the maximal height) \cite{B06}. It follows from Proposition \ref{prop: weight for vertices} that the diameter $D_G$ of the tiling graph $G$ satisfies: $D_G \leq N^\frac{3}{2}$, where $N = \vert V_R \vert$.

Define a Markov chain on $\Omega$ as follows:

\subsection*{MC:}

Let $T_0$ be an initial configuration. At each time $t$:
\begin{enumerate}
\item choose an inner vertex of $R$ with probability u.a.r.,
\item choose a direction of a flip with probability $\frac{1}{2}$,
\item perform a flip in the chosen direction in the tiling $T_t$ if possible thus defining the tiling $T_{t+1}$, otherwise stay still.
\end{enumerate}

\begin{lemma}
$MC$ has uniform stationary distribution.
\end{lemma}

\begin{proof}
$MC$ is \textit{irreducible} since any tiling from $\Omega_R$ can be obtained from any other tiling via a finite number of consecutive flips, it is \textit{aperiodic} since the self-loop probability is greater than zero. Therefore, the chain is \textit{ergodic} and it converges to its unique stationary distribution. Moreover, the transition probabilities are symmetric, so the stationary distribution is uniform.
\end{proof}

\subsection*{Coupling of general Markov chain}
\label{subsec: coupling_general}

We construct a coupling for  $MC$. Let $A_0$ and $B_0$ be two starting configurations. At time $t$ we pick a vertex from $V_R$ and a direction (same for both tilings). We make a flip in $A_t$ and $B_t$ if possible; this defines $A_{t+1}$ and $B_{t+1}$.  This randomizing process respects the order: that is if $h(B_t) < h(A_t)$, then $h(B_{t+1}) \leq h(A_{t+1})$. Simulations using Coupling from the Past for the $MC$ suggest that its mixing time $\tau_{mix} = O(N^{2.5})$.
\begin{figure}[hbtp] 
\centering
\includegraphics[width=0.8\textwidth]{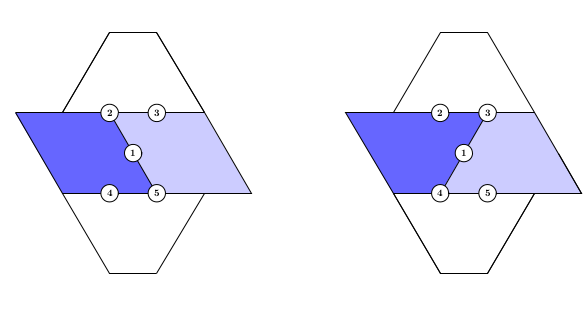}
\caption{Bad configuration for path coupling.}
\label{fig: bad_path_coupling}
\end{figure}

We wish to apply the path coupling method (see Theorem \ref{theorem: DG}) to bound $\tau_{mix}$. For tilings $A, B \in \Omega_R$ define the \textbf{distance} function  $\varphi$ on $\Omega_R \times \Omega_R \rightarrow \mathbb{Z}$ as
$$
\varphi^\ast(A, B): = h(A) - h(B), \, \varphi := \frac{\varphi^\ast }{3}.
$$
Let $U$ be a subset of $\Omega_R \times \Omega_R$  that consists of pairs of configurations that differ by one flip. We wish to prove  that for all $A_t$, $B_t \in U$ 
$$
\mathbb{E}[\Delta{\varphi(A_t, B_t)}] \leq 0,
$$
where $ \Delta{\varphi(A_t, B_t)} = \varphi(A_{t+1}, B_{t+1}) - \varphi(A_t, B_t). $

Unfortunately, with such definition of the distance, the inequality does not stand. The following example illustrates that. Consider two tilings $A$ and $B$ that differ by one flip around one vertex. See Figure \ref{fig: bad_path_coupling} for an illustration, where tilings are different in vertex 1. Then a flip around vertex 1 gives $-\frac{1}{N}$ in terms of expectation, and $+\frac{1}{2N}$ around each one of vertices $2-5$. Putting it all together, one gets:

\begin{equation}
\label{eq: bad_coupling}
 \mathbb{E}[\Delta{\varphi}] = -\frac{1}{N} + \frac{1}{2N} + \frac{1}{2N}+ \frac{1}{2N}+ \frac{1}{2N} > 0.
\end{equation}

This means that either a better metric should be found (for example, what Wilson did with the lozenge tilings \cite{W04}) or the Markov chain should be modified in such a way that would make it possible to analyse it (one can think of the tower of flips for domino and lozenge tilings by Luby, Randall, Sinclair \cite{LRS95}). In the next two subsections we consider a weighted version of the Markov chain such that the coupling method works, and we consider as well restrained Kagome tilings and prove that the corresponding Markov chain is fast mixing for specific regions. 

\section{Weighted Glauber dynamics}
\label{sec: weighted chain}

A popular approach in the statistical physics is to consider weighted models (\textbf{e.g.} dimer model, Ising model etc). It turns out that putting weights on certain configurations can make the analysis of defined Markov chains significantly easier. Let us follow the same approach for Kagome tilings. We change probabilities of flips in the following way. Consider Figure \ref{fig: all_flips}: notice that there are three flips out of six that change the number  of \textit{fish} tiles -- each of these flips makes a fish tile appear or disappear. Let us refer to them  as \textbf{fish changer} flips, and to all others -- \textbf{fish-stable} flips (these flips do not change the number of fish tiles). Let us modify initial $MC$ by putting weights on the fish changer flips.

\subsubsection*{MC$_{\lambda}$ with weight $\lambda >0$:}

Let $T_0$ be an initial configuration. At each time $t$:
\begin{itemize}
\item choose an inner vertex of $R$ u.a.r.;
\item 1. If a fish changer flip can be made in the tiling $T_t$, make it with probability $\frac{1}{1+\lambda}$ if it makes a fish tile disappear and with probability  $\frac{\lambda}{1+\lambda}$ if it makes a fish tile appear. This defines the tiling $T_{t+1}$;\\
2. Otherwise, if possible, perform a fish-stable flip in a random direction in $T_t$, thus defining $T_{t+1}$ (as in the unweighted version); if not, stay still.
\end{itemize}

\textbf{Remark.} If $\lambda$ is small (less than 1), fish changer flips that increase the number of fish tiles are being penalized. The intuition behind this penalization is based on simulations that suggest that its the fish tiles that slow down the chain's mixing time. When $\lambda = 1$, one gets the initial $MC$.

\begin{figure}[hbtp] 
 \hspace*{-5mm}
\includegraphics[width=1\textwidth]{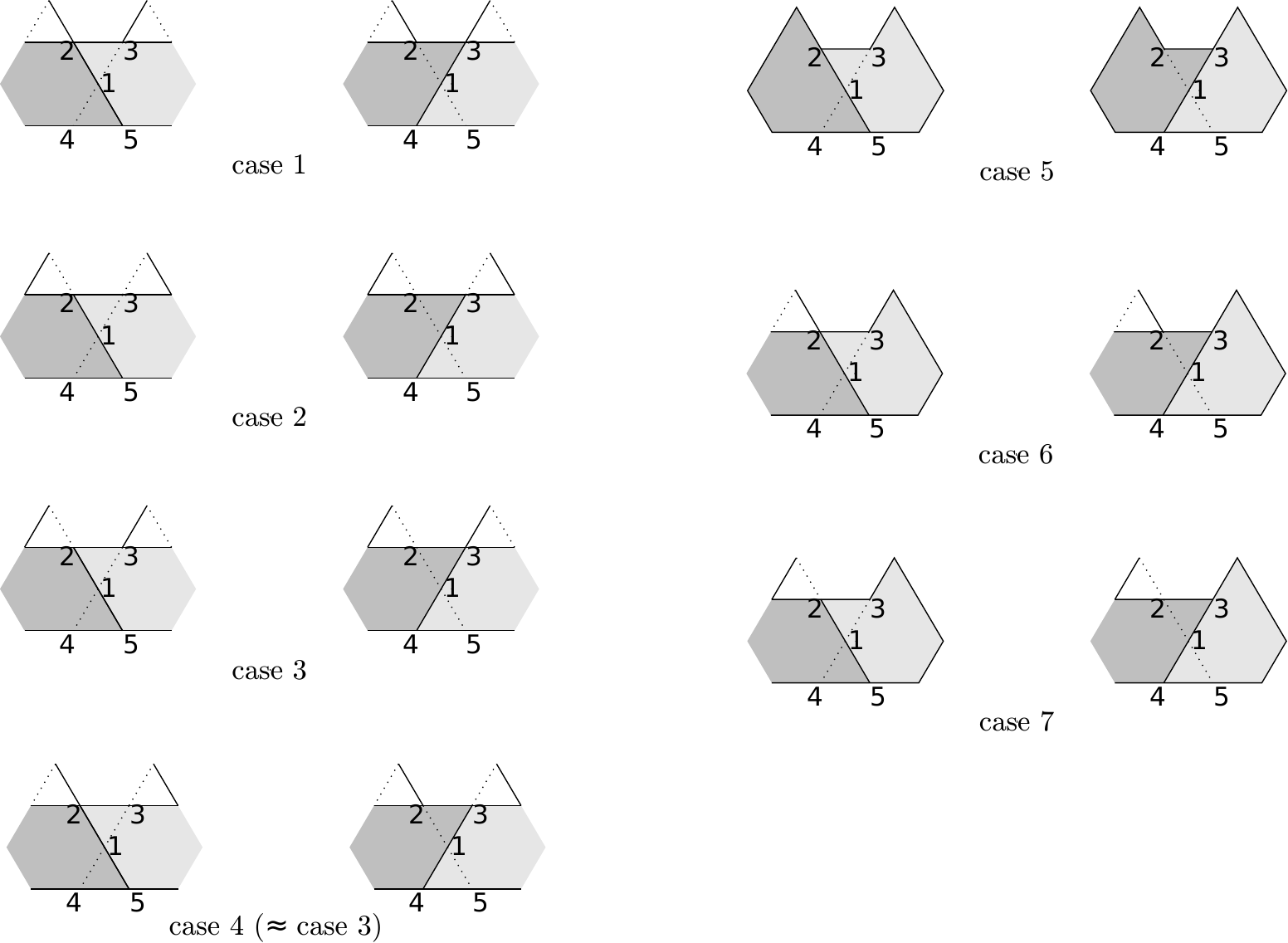}
\caption{All possible cases for a pair of configurations at distance 1 are shown schematically.}
\label{fig: kagome_all_cases}
\end{figure} 

\begin{theorem}\label{thm: weighted kagome mixing}
Consider Kagome tilings of a finite simply connected region $R$ of area $N$ of the Kagome lattice. $MC$ with weight $\lambda$ is rapidly mixing for $\lambda \leq \frac{1}{3}$. The following bound stands for $\tau_{mix}$ and some positive constant $c$:
\begin{equation}\label{eq:weighted bound}
\tau_{mix} \leq c N^4.
\end{equation}
\end{theorem}

\begin{proof}
Consider a pair of tilings $(A,B)$ different by one flip. One can think of Figure \ref{fig: bad_path_coupling} from the previous subsection. All possible configurations are shown in Figure \ref{fig: kagome_all_cases}. One observation is that if the vertex $1$ is chosen, flips around it will decrease  the distance by $1$ in any of the case: in terms of expectation one will have either  $-\frac{1}{2} - \frac{1}{2}$ if a fish free flip is made or $ -\frac{1}{\lambda + 1 } - \frac{\lambda}{\lambda + 1 } $ if a fish changer flip is made.

Another observation is that there are at most 4 \textbf{bad} vertices  -- such vertices that increase the distance between $A$ and $B$. The only layout possible to have 2 bad vertices on each side (vertices 2,3 and 4,5) is to have a trapeze tile sharing its longest side with the two tiles in question as is shown in Figure \ref{fig: bad_path_coupling}. In this case a flip around any of the vertices 2,3,4 and 5 that can increase the distance is the fish changer flip that makes a fish tile appear. So it is performed with probability $\frac{\lambda}{\lambda + 1 }$. The equation \eqref{eq: bad_coupling} that failed mixing now turns into:

\begin{equation}
\label{eq: weighted_coupling1}
\mathbb{E}[\Delta{\varphi}] = -\frac{1}{N} +  \frac{1}{N} \frac{\lambda}{\lambda + 1 } + \frac{1}{N}\frac{\lambda}{\lambda + 1 }+ \frac{1}{N}\frac{\lambda}{\lambda + 1 }+ \frac{1}{N}\frac{\lambda}{\lambda + 1 }.
\end{equation}

We want \eqref{eq: weighted_coupling1} to be $\leq 0$, so one gets : 
$$
\frac{4 \lambda}{\lambda + 1} \leq 1,
$$
which yields the bound on $\lambda$
$$
\lambda \leq \frac{1}{3}.
$$

We are pretty much done with the proof, since the case with one bad vertex on each side gives, in terms of expectation, at most: $-\frac{1}{N} + \frac{1}{2N}+ \frac{1}{2N} = 0$. And the mixed situation with one bad vertex on one side and two bad vertices on the other side gives, in terms of expectation, at most: $-\frac{1}{N} + \frac{1}{2N}+ \frac{1}{N}\frac{\lambda}{\lambda + 1 }+ \frac{1}{N}\frac{\lambda}{\lambda + 1 }.$ This is surely less than $0$ when $\lambda \leq 1/3 $.

Since  $\mathbb{P} (\varphi(A_{t+1},B_{t+1}) \neq  \varphi(A_t,  B_t) \vert A_t,B_t) \geq \frac{1}{N}$, the coupling theorem gives

$$
\tau_{mix} \leq c \frac{D_G^2}{1/N},
$$
where  $D_G$ is the diameter of the tiling graph $G_{tiling}^R$. So one gets \eqref{eq:weighted bound}.
\end{proof}

\section{Restrained Kagome tilings}
\label{sec: restrained}

We have already seen the previous section that since it seems that the fish tiles are the ones that slow down the mixing, we can just penalize fish tiles in tilings. This allows us to get rapid mixing with weights $\lambda \leq \frac{1}{3}$. A different approach is to get rid of fish tiles completely and consider Kagome tilings where the set of prototiles is restrained to a trapeze and a lozenge.

Only flips that contain these two prototiles are allowed -- see Figure \ref{fig: partial_flips1}. We refer to these flips as \textbf{restrained} flips. A vertex $v$ is \textbf{flippable} if a restrained flip can be perform around it. An inner vertex $v$ is called a \textbf{local flippable maximum/minimum} if $v$ is a local maximum/minimum and a restrained flip can be perform around it that decreases/increases $h(v)$ by $3$. It is called \textbf{non-flippable} otherwise. A flippable local maximum/minimum is always a local maximum/minimum (as defined for general Kagome tilings), while the opposite is not true. A non-flippable local maximum $v$ can never be a global maximum: if $h(v)= M$, then one of its neighbours with height $M-1$ belongs to a triangle where there is a vertex with height $M+1$.

Note that partial order is still respected. It means that for a pair of tilings $T_1, T_2$ of $R$, if $h(T_1) \geq h(T_2)$, then after performing a restrained flip in vertex $v$ in both tilings (if possible), we have that $h(T'_1) \geq h(T'_2)$.

\begin{figure}[hbtp] 
\centering
\includegraphics[width=0.6\textwidth]{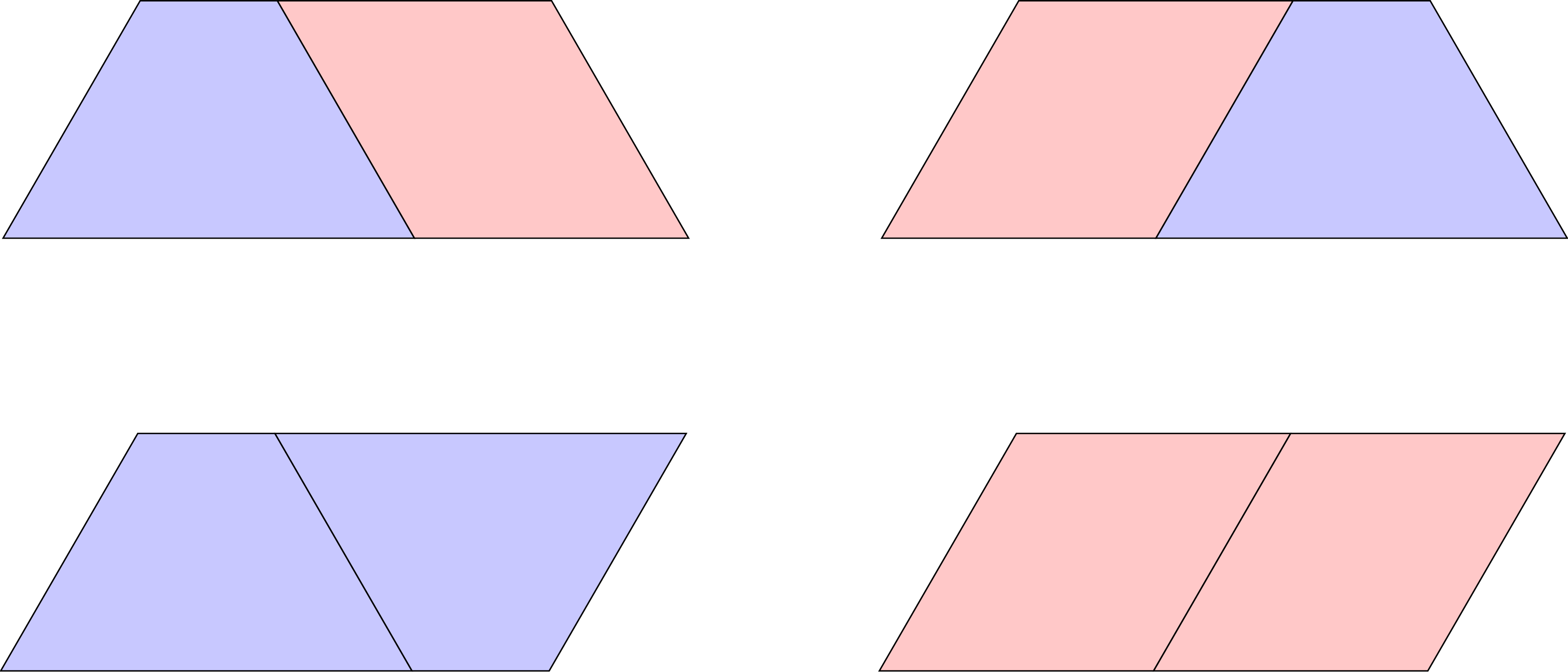}
\caption{Possible flips in restrained Kagome tilings.}
\label{fig: partial_flips1}
\end{figure} 



Define a Markov chain $MC_{restr}$ in the same way as the general $MC$ where only restrained flips are performed. Throughout this part we consider that the region $R$ is a lozenge with sides of length $2n$. We say that it is of size $N$, where $N = 2n$. Denote the set of all such possible restrained Kagome tilings of $R$ by $\Omega$ as before.  A tiling $T \in \Omega$ is called \textbf{minimal} if $T$ does not have local flippable maxima in $R \setminus \partial R$. Let $\Omega_{min}$ be the set of all minimal tilings for $R$, $\Omega_{min} \subset \Omega $.

\begin{theorem}
Let $R$ be a tileable lozenge region of the Kagome lattice and $\Omega$ be the set of all possible restrained Kagome tilings of $R$. Then $\Omega$ is connected.
\end{theorem}

\begin{proof}

The connectivity follows from the following facts:
\begin{enumerate}

\item For $R$ there exists a minimal tiling $T_{m} \in \Omega_{min}$;
\item For each $T \in \Omega$ there exists a unique tiling $T_{min} \in \Omega_{min}$ achievable from $T$ by height-decreasing flips;
\item The height $h$ on $\partial R$ depends only on $R$ and not on the tiling;
\item A tiling in $\Omega_{min}$ is completely determined by the height on $\partial R$.
\end{enumerate}

Let us point out that the first three facts are standard and work for all finite tileable simply connected regions. It is the fourth one which requires a more thorough approach specific to the lozenge region. 

1. Consider a tiling $T_m \in \Omega$. Suppose that for all tilings $T \in \Omega $ $h(T_m) \leq h(T)$. Then let us show that $T_m$ is a minimal tiling. Suppose there exists a vertex $v_m$ which is a flippable local maximum for $T_m$ and it is in $R \setminus \partial R$. Then there exists a tiling $T \in \Omega$, such that $T_m$ will be transformed into $T$ after performing a flip around $v_m$. Moreover, this flip reduces $h_{T_m}(v_m)$ by $3$. Then we have that $h(T) = h(T_m)-3$, which contradicts the hypothesis of $T_m$ having the minimal height. This means that tilings that have minimal height are minimal.

2. Consider a tiling $T_o \in \Omega$. Start performing only height decreasing flips while it is possible. A tiling $T_{min}$ obtained in the end is clearly minimal. Note that the order of flips does not matter since two maxima can not be neighbours and a flip in a vertex cannot block a flip in the same direction in neighbouring vertices (vertices that belong to the same triangle as the chosen vertex). Therefore, for each tiling $T_o$ there exists a unique tiling $T_{min} \in \Omega_{min}$ obtained by the sequence of height decreasing flips as described above. 

3. It is clear that $h$ on $\partial R$ is defined by flow and orientation as for general Kagome tilings. 

4. Let us show that a minimal tiling $T_{min} \in \Omega_{min}$ is completely defined by the height of the boundary of $R$. 
The idea is the following: choose cells of the lozenge $R$ in a specific order and show that at each step the way of placing a tile is unique, which means that $T_{min}$ is completely defined by $h$ on $\partial R$.

We use induction to show that the minimal tiling of a lozenge of size $N$ is uniquely reduced to a lozenge of size $N-4$. 

The base of induction is $N = 1,2,3$. It is not difficult to see that there is a unique way of placing the tiles to get a minimal tiling. Minimal tilings for $N = 1,2,3$ are shown in the Figure \ref{fig: k10}.

\begin{figure}[hbtp] 
\centering
\includegraphics[width=1\textwidth]{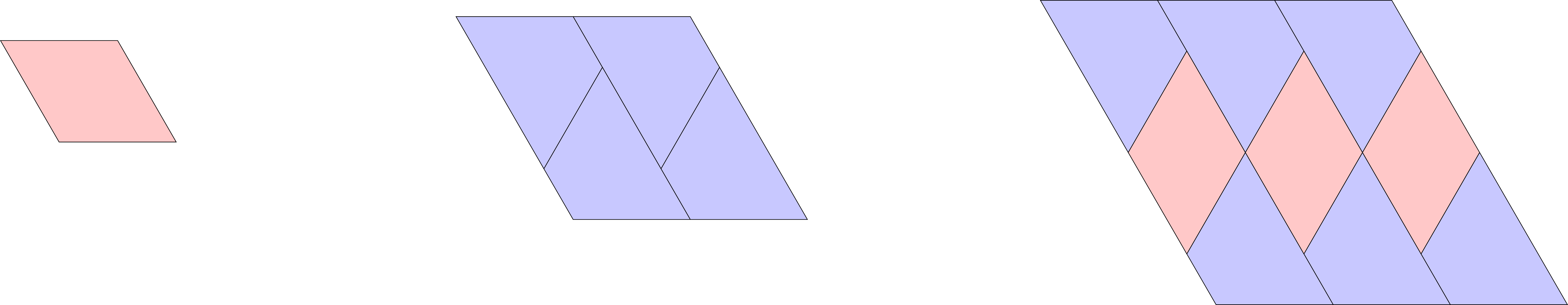}
\caption{Minimal tilings of lozenges of size $N, N = 1,2,3$.}
\label{fig: k10}
\end{figure}

Let us show that there is a unique was of defining a part of $T_{min}$ of size $N$ that brings us to a lozenge of size $N-4$ for which $T_{min}$ is uniquely defined by induction. 

We consider a tile in the upper left corner of the lozenge. There are three possible ways to place a tile, two of which create a local maximum marked by a circle:

\begin{figure}[h!] 
\centering
\includegraphics[width=1\textwidth]{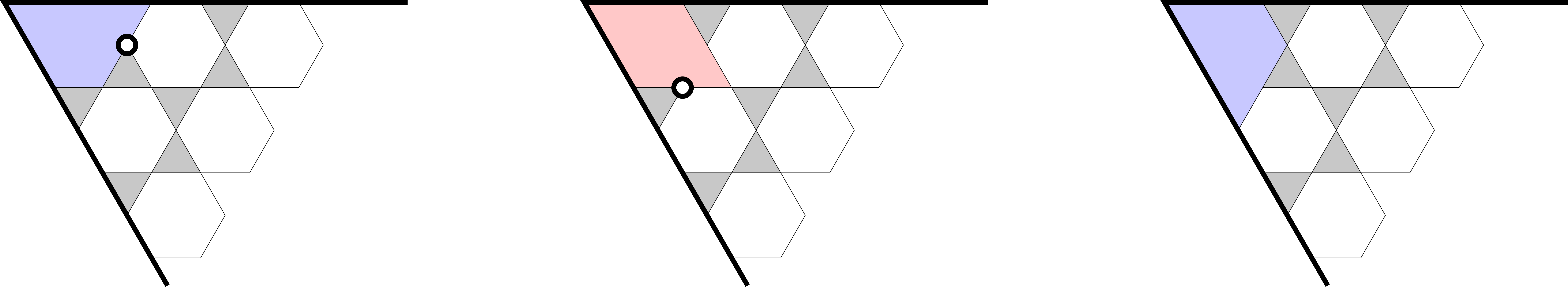}
\end{figure} 

\newpage
The local maximum has to be non-flippable otherwise the tiling is not\\ minimal. The first case forces to place a trapeze on the right which creates a hole below that cannot be completed:

\begin{figure}[h!] 
\centering
\includegraphics[width=0.27\textwidth]{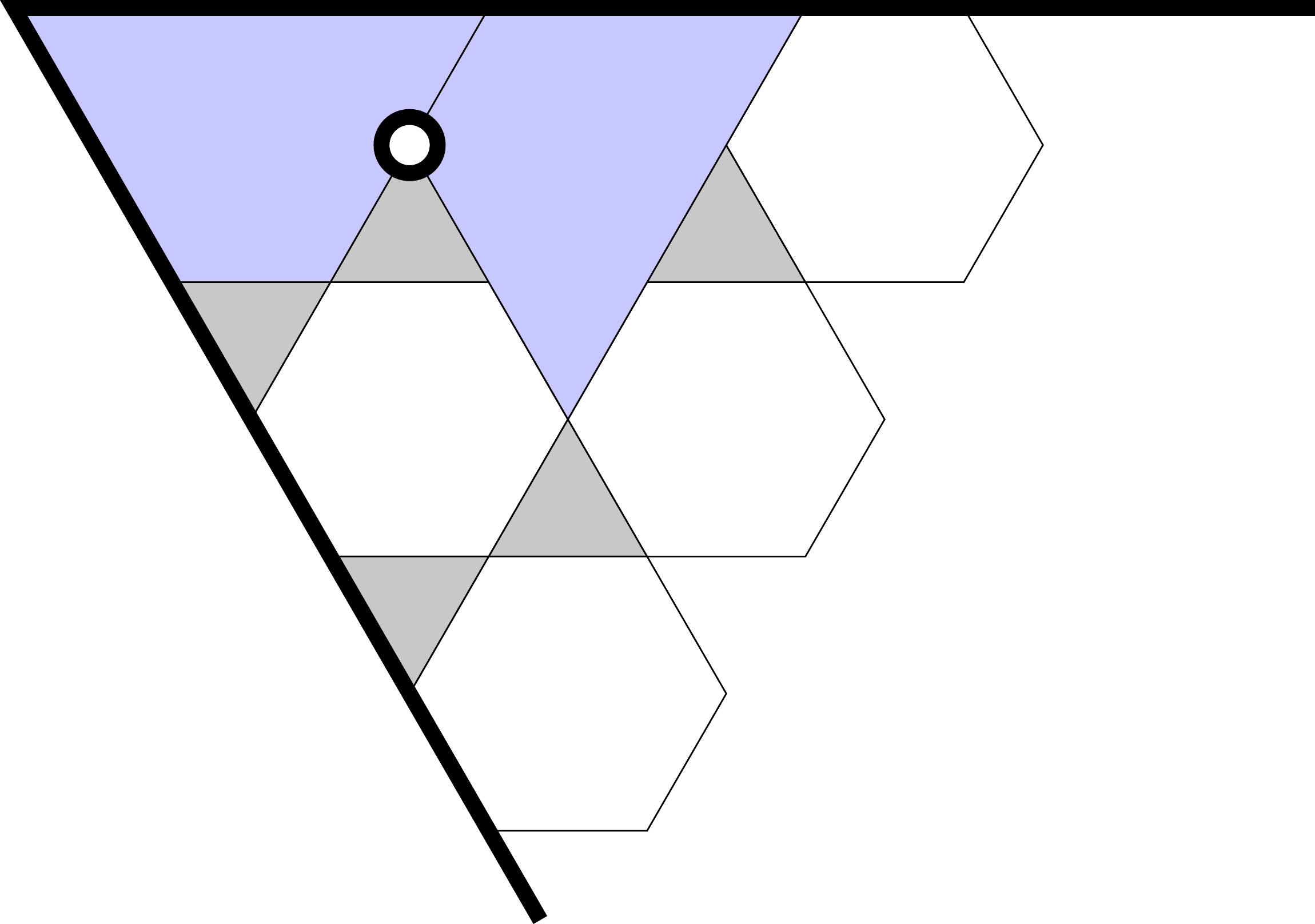}
\end{figure}

The second case forces a trapeze below and a new local maximum. The local maximum has to be non-flippable, so it leaves three possibilities on the right, each creating a third local maximum:

\begin{figure}[h!] 
\centering
\includegraphics[width=1\textwidth]{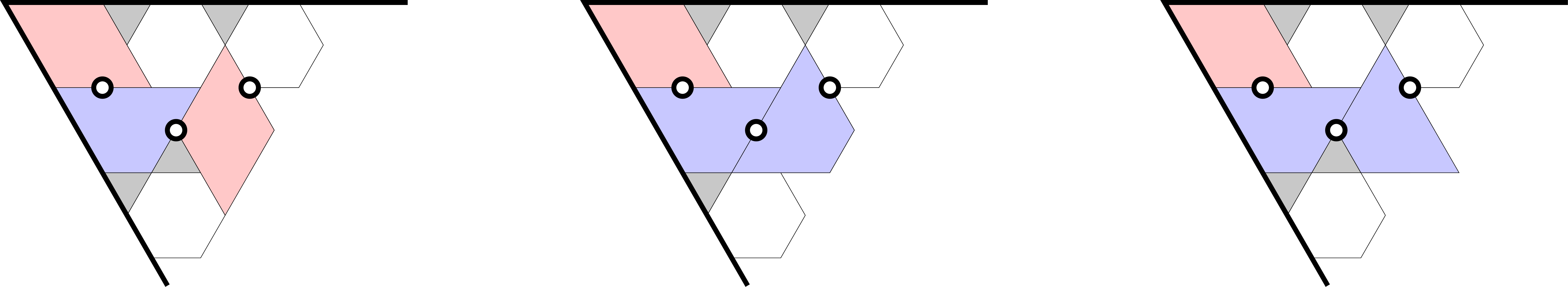}
\end{figure}

Whichever the tile from the picture above is, it occupies the same triangle and hexagon (dashed in the picture below) and create a local maximum. This local maximum has to be non-flippable, which once again allows three possible tiles to be placed, all of them creating a hole that cannot be completed: 

\begin{figure}[hbtp] 
\centering
\includegraphics[width=1\textwidth]{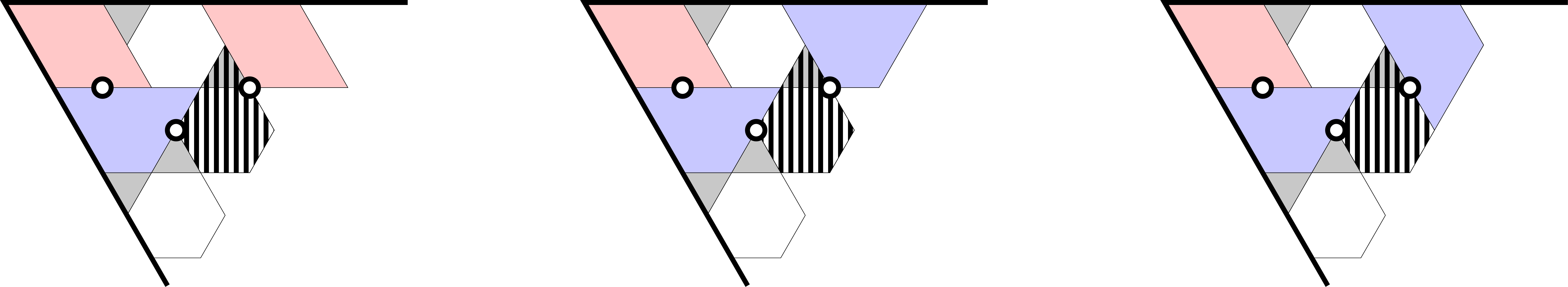}
\end{figure} 

Therefore, the first tile is uniquely defined. We proceed with the first line from left to right. There are three possibilities for the second tile: 

\begin{figure}[h!] 
\centering
\includegraphics[width=1\textwidth]{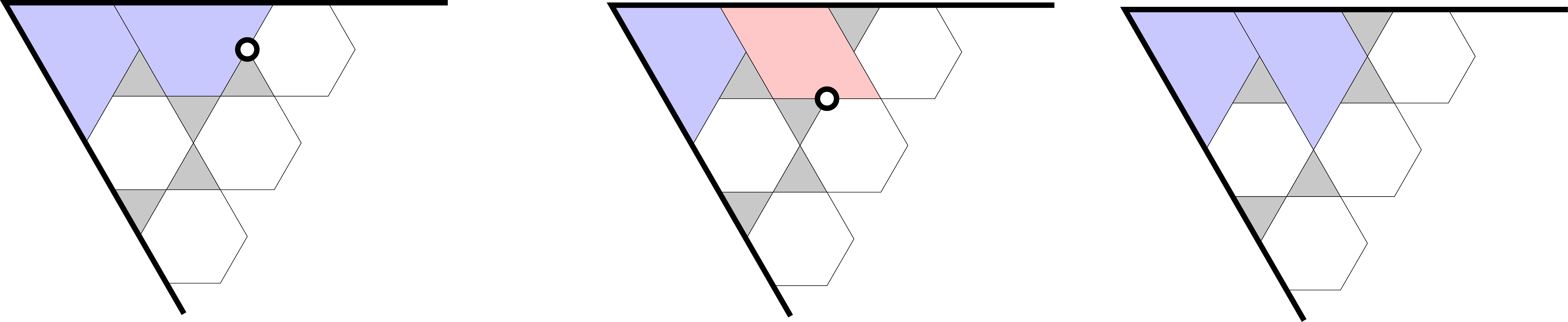}
\end{figure} 

\newpage

In the first case, a trapeze is forced. The marked triangle cannot be covered without placing a fish tile: 

\begin{figure}[h!] 
\centering
\includegraphics[width=0.27\textwidth]{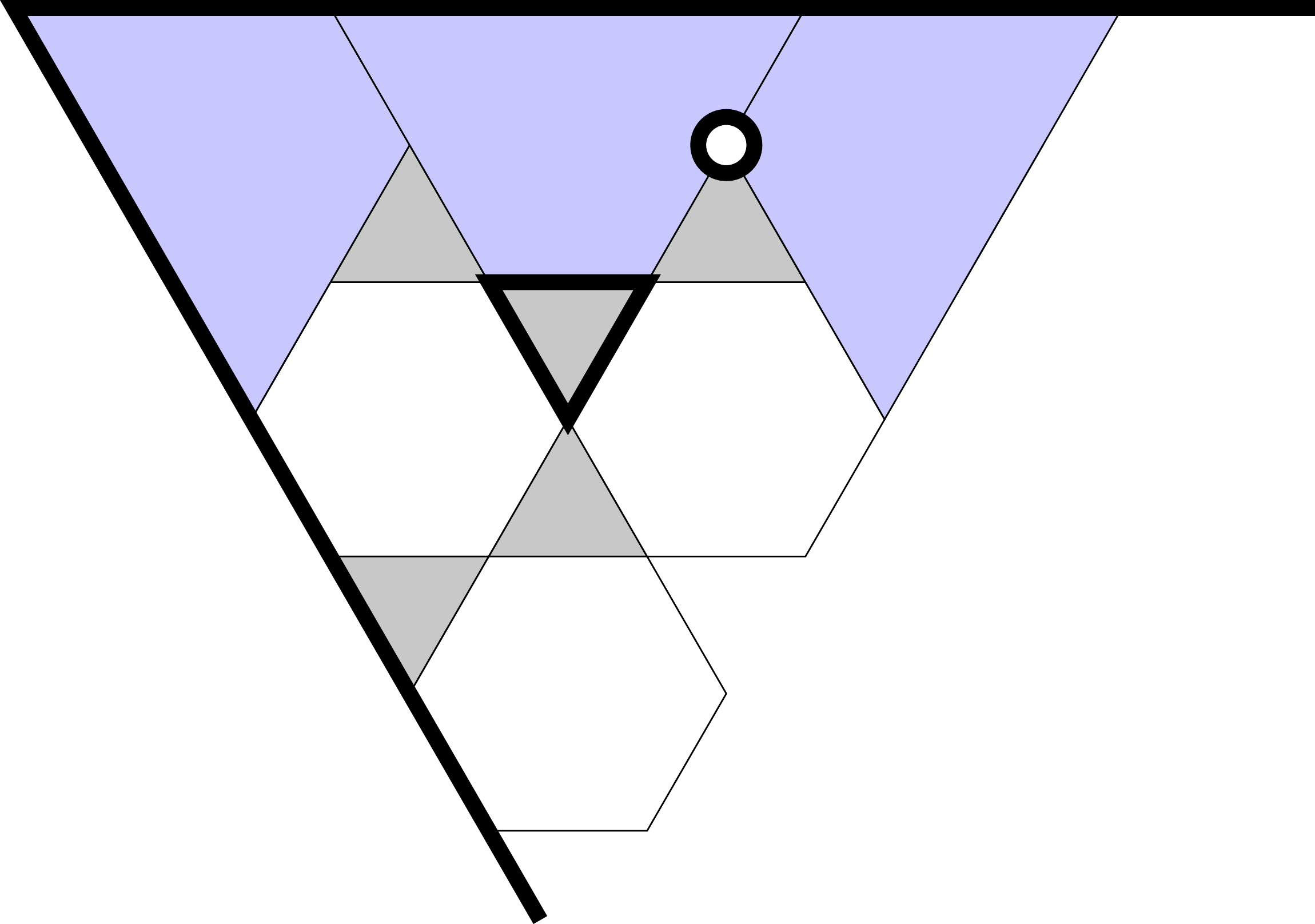}
\end{figure}

In the second case there is a local maximum that cannot be flippable, which brings us to a case above already excluded.

\begin{figure}[h!] 
\centering
\includegraphics[width=0.27\textwidth]{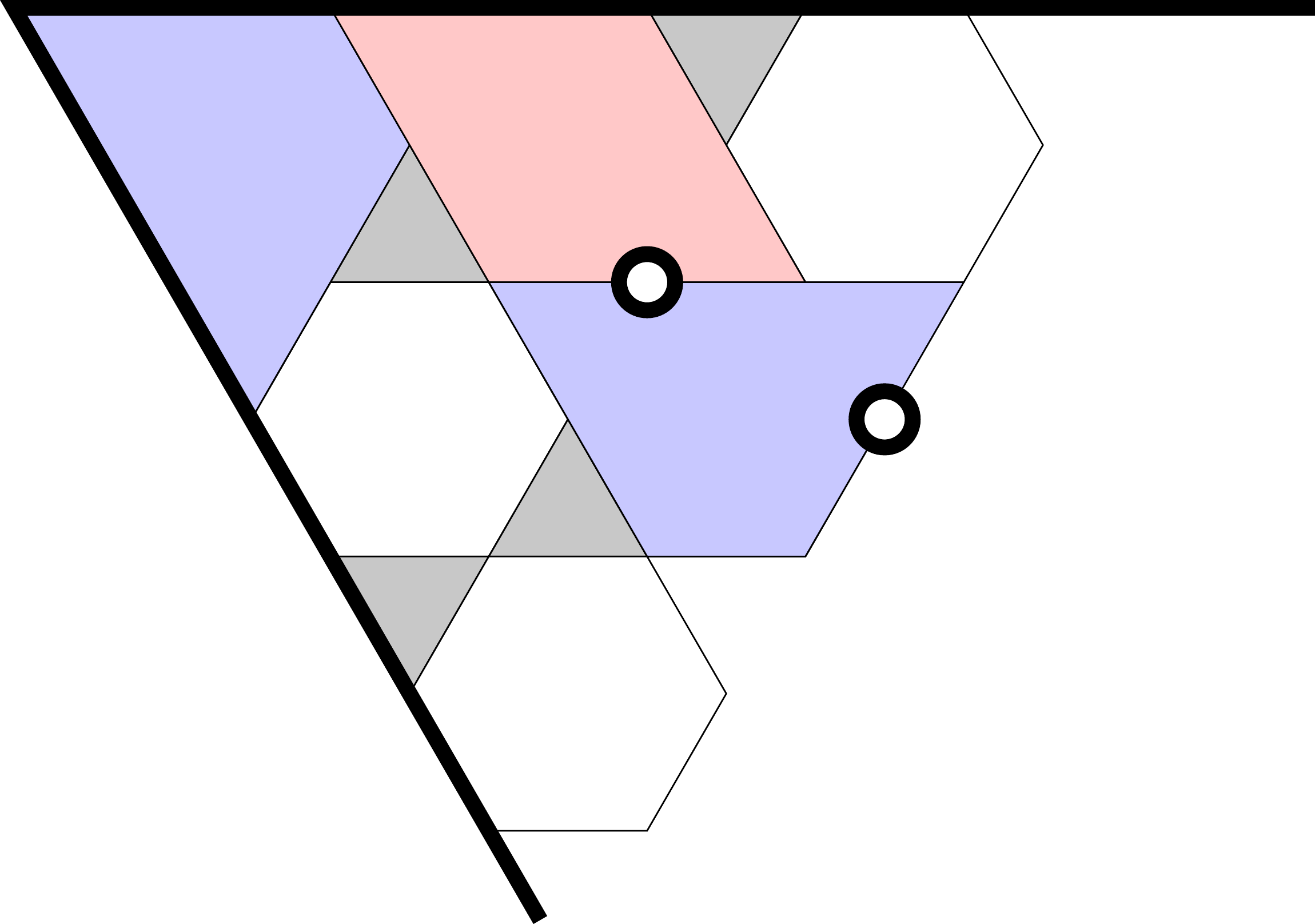}
\end{figure} 

Following the same argument, the first line of the lozenge is uniquely defined:

\begin{figure}[h!]
\centering
\includegraphics[width=0.27\textwidth]{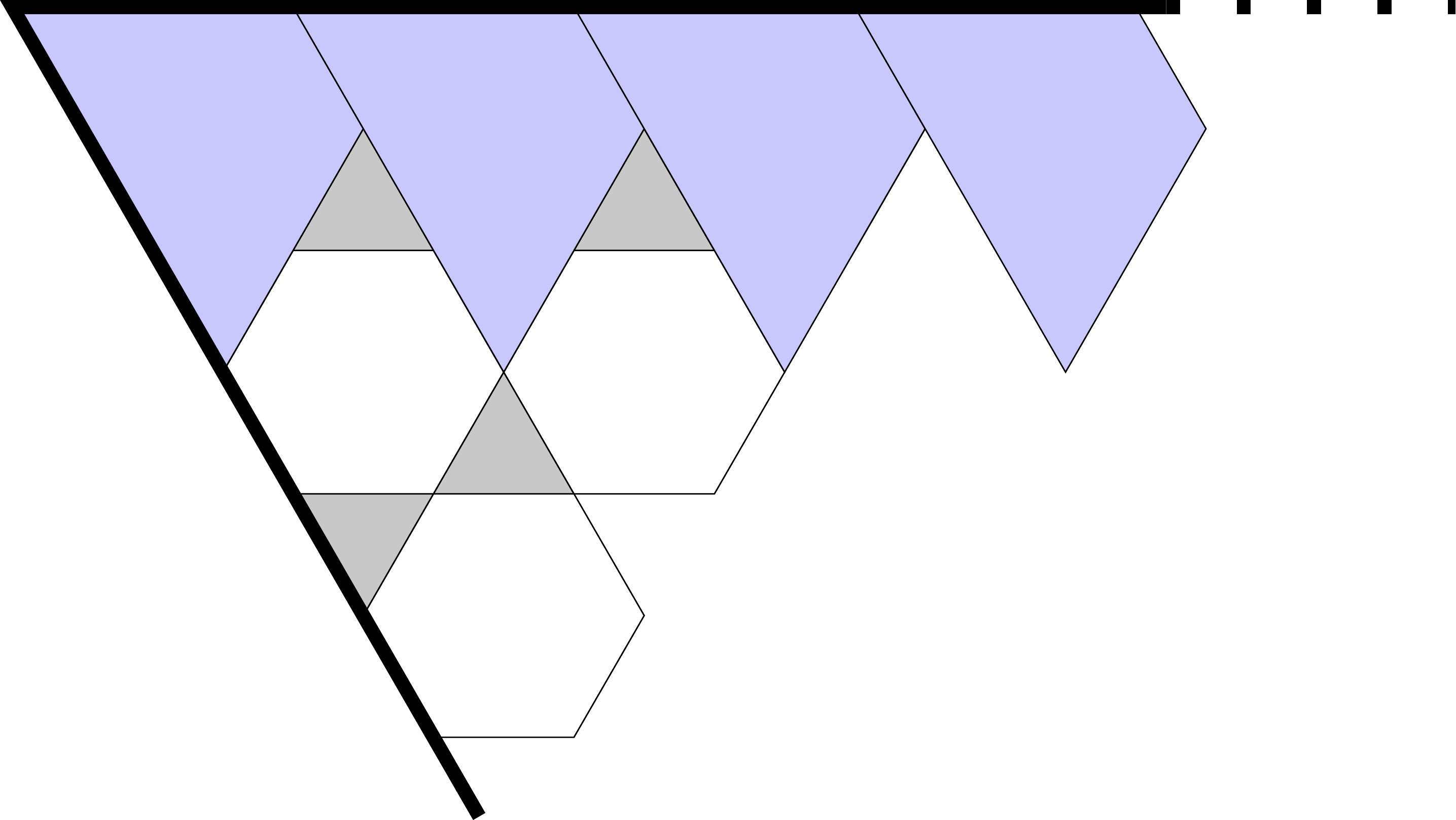}
\end{figure}

Let us now move to the second line following the left-to-right order once again. There are two choices for the first tile in the second line:

\begin{figure}[h!] 
\centering
\includegraphics[width=0.6\textwidth]{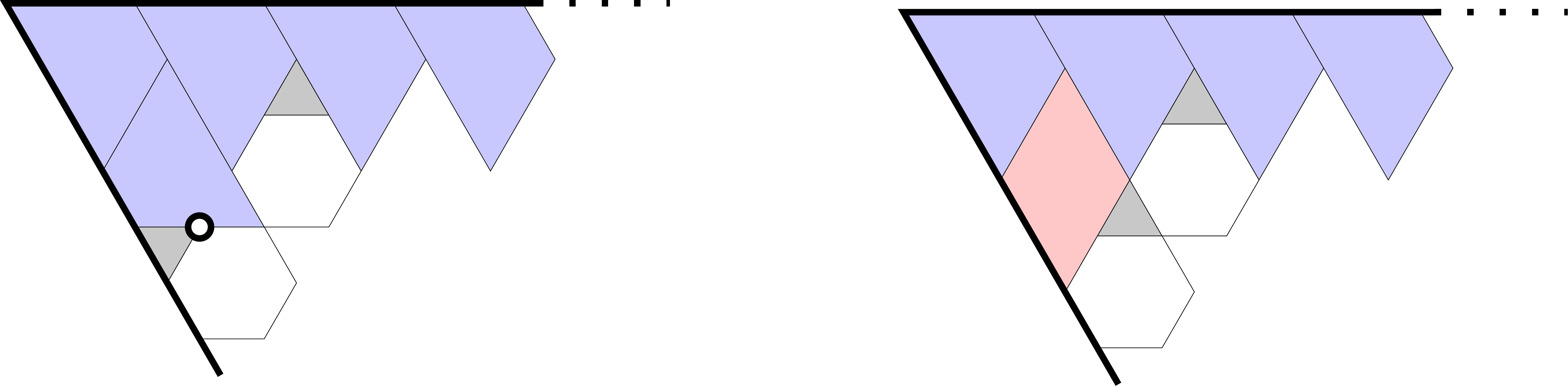}
\end{figure}

The first case creates a local maximum that has to be non-flippable. This forces a trapeze below that makes the hexagon on the right be coupled with the triangle above (dashed in the picture below) and creates a hole that cannot be completed:

\begin{figure}[h!] 
\centering
\includegraphics[width=0.27\textwidth]{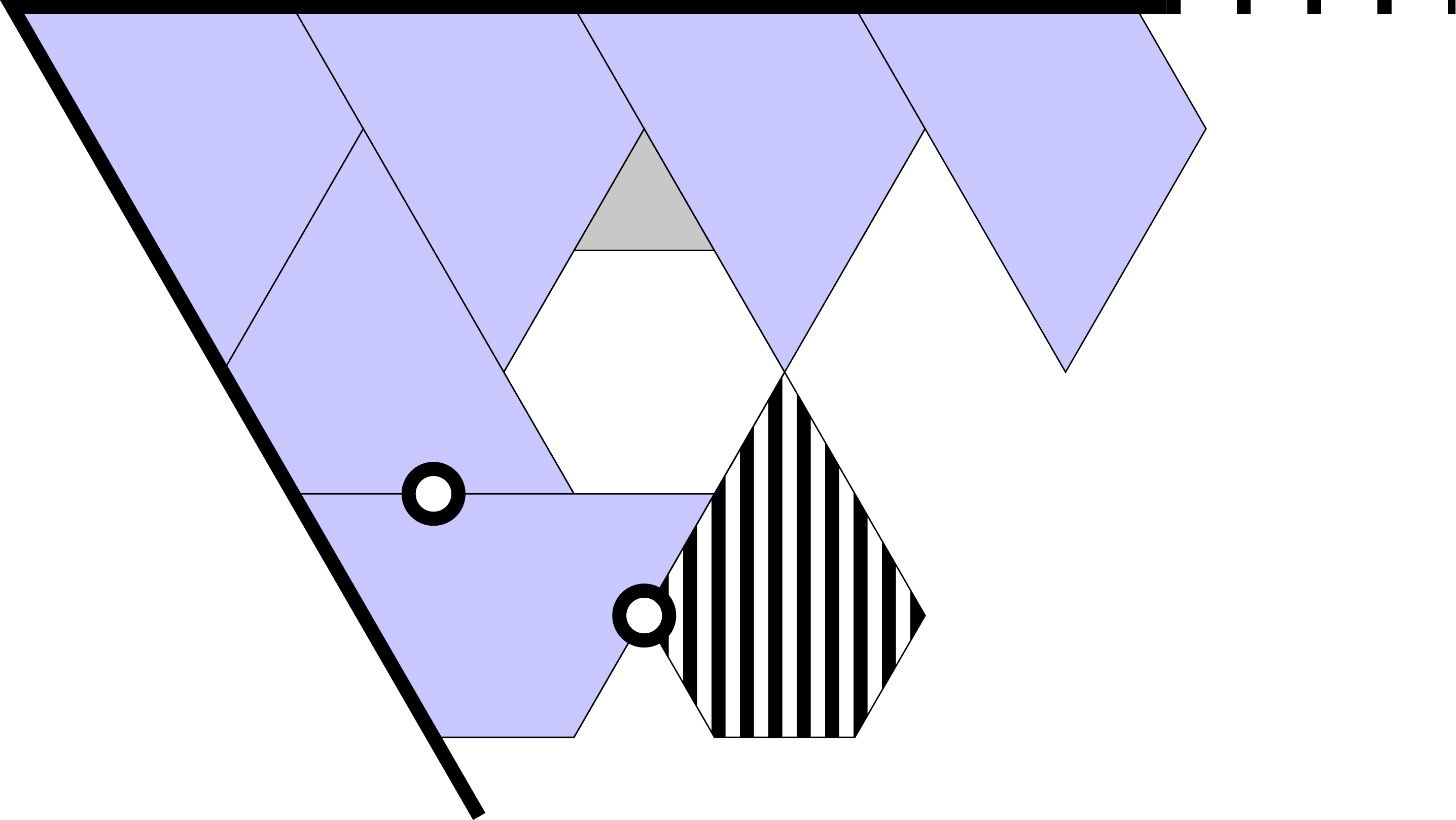}
\end{figure}

\newpage

There are three possibilities for the second hexagon cell of the second line:
\begin{figure}[h!] 
\centering
\includegraphics[width=1\textwidth]{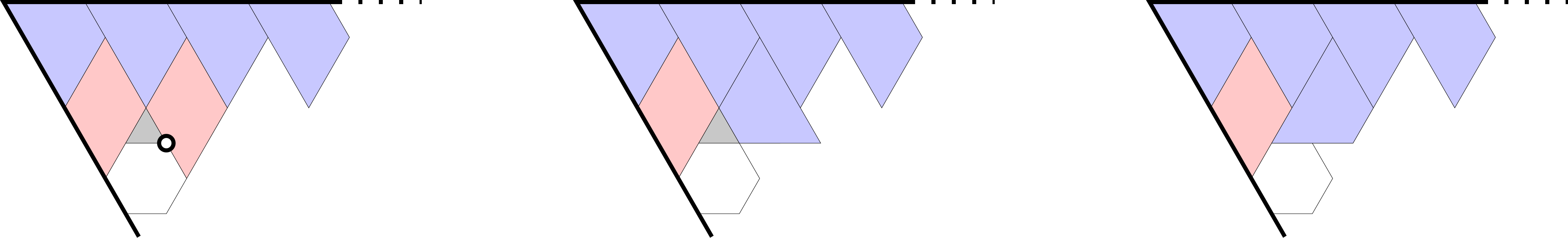}
\end{figure} 

The first two cases are excluded due to the similar reasoning as before (the created local maxima have to be non-flippable which leads to a hole that cannot be completed):

\begin{figure}[h!] 
\centering
\includegraphics[width=0.6\textwidth]{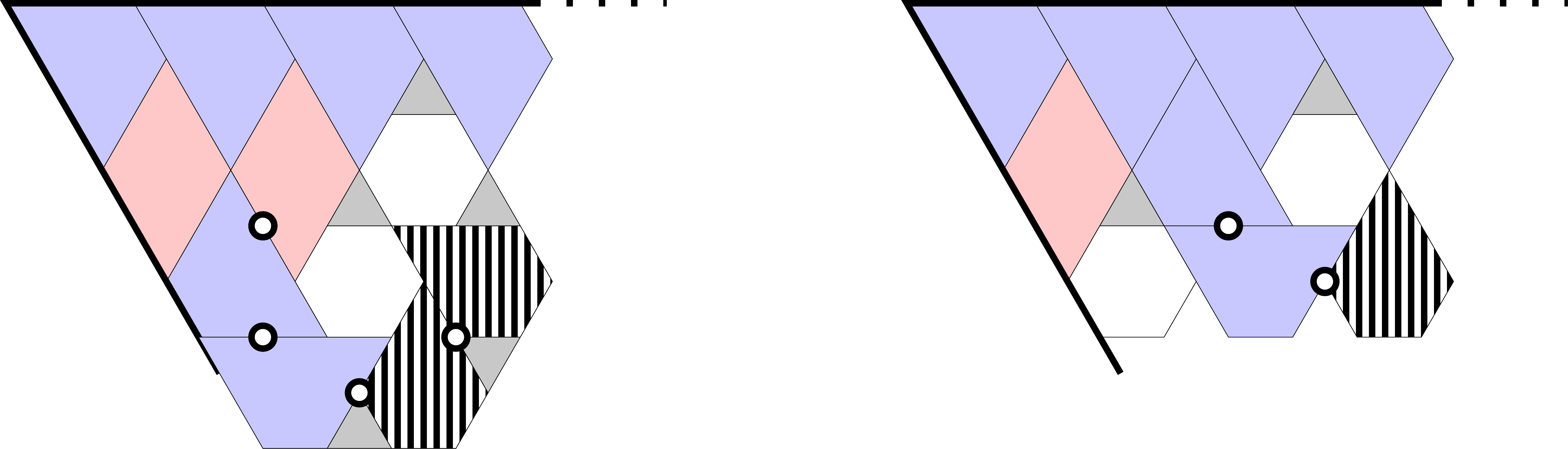}
\end{figure} 

Following the same argument, the second line of the lozenge is uniquely defined:
\begin{figure}[h!] 
\centering
\includegraphics[width=0.57\textwidth]{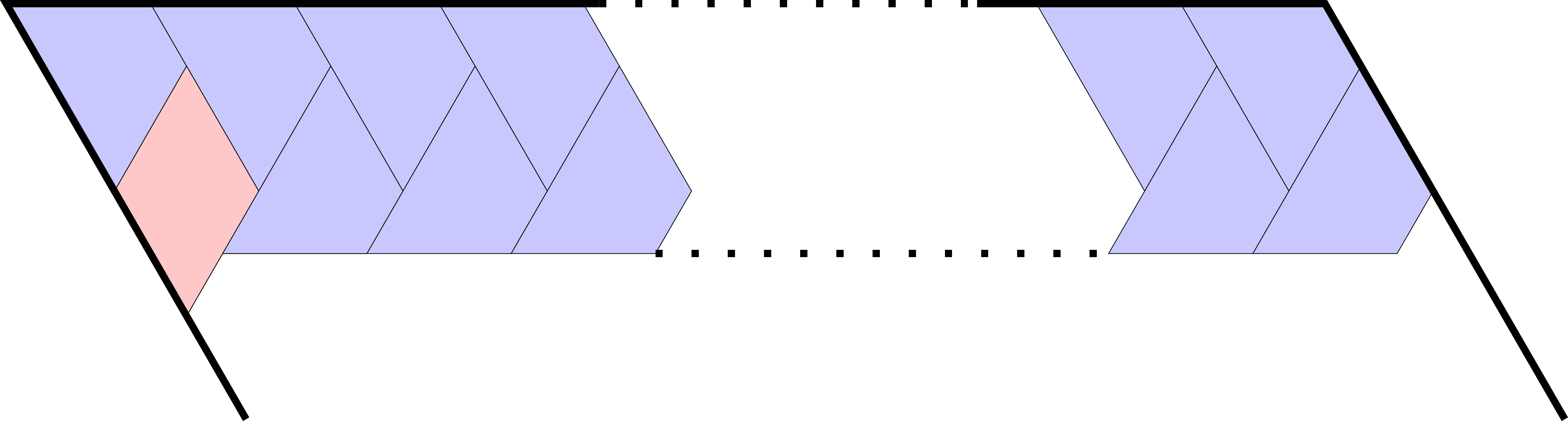}
\end{figure}   

We do the same thing for the two bottom lines of the hexagon in the right-to-left direction. This gives the following:

\begin{figure}[h!] 
\centering
\includegraphics[width=0.7\textwidth]{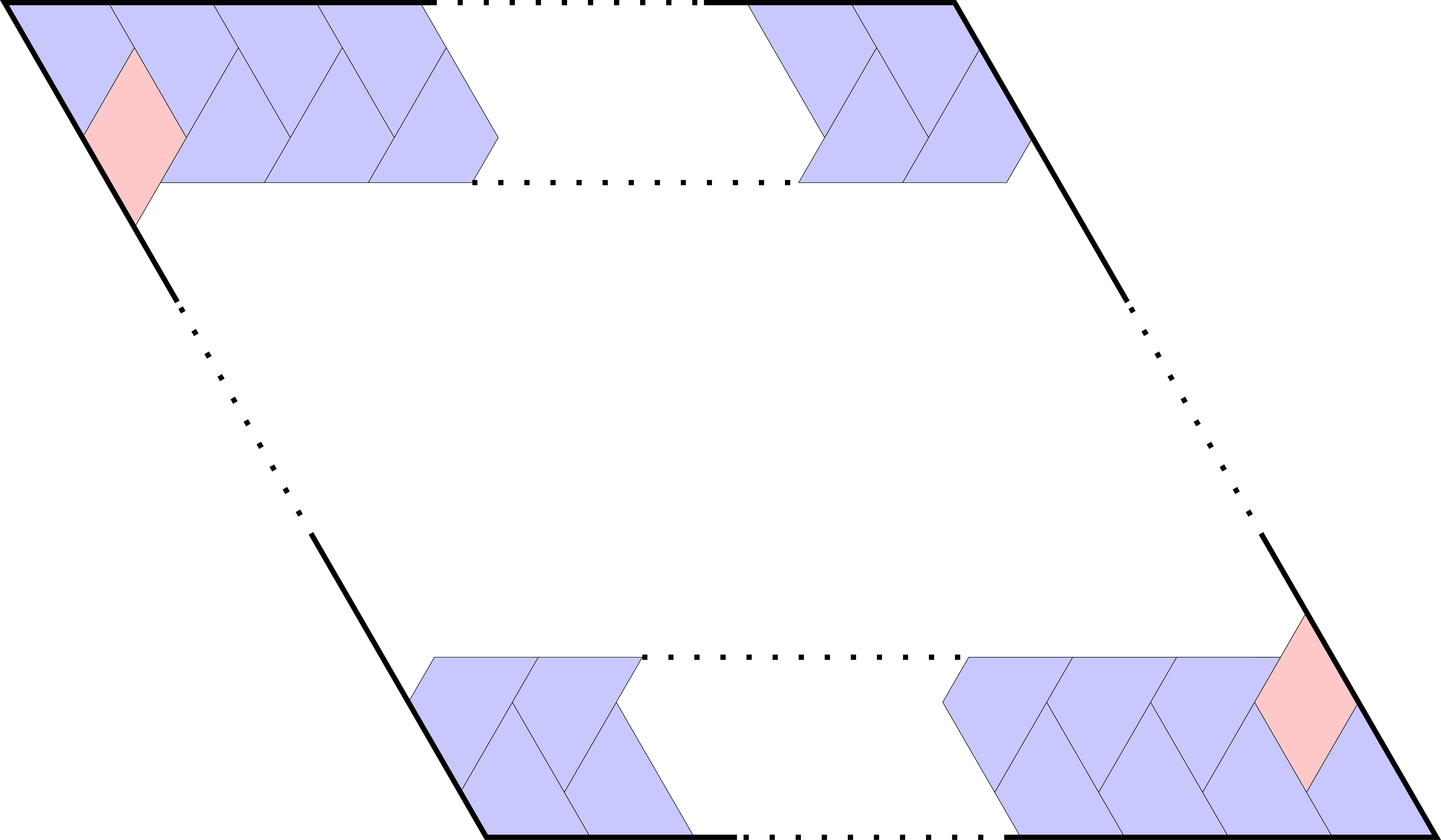}
\end{figure}  
\newpage

Using the same reasoning, tiles are uniquely defined one by one in the up-to-down direction for the leftmost column, and then for the second column:

\begin{figure}[h!] 
\centering
\includegraphics[width=1\textwidth]{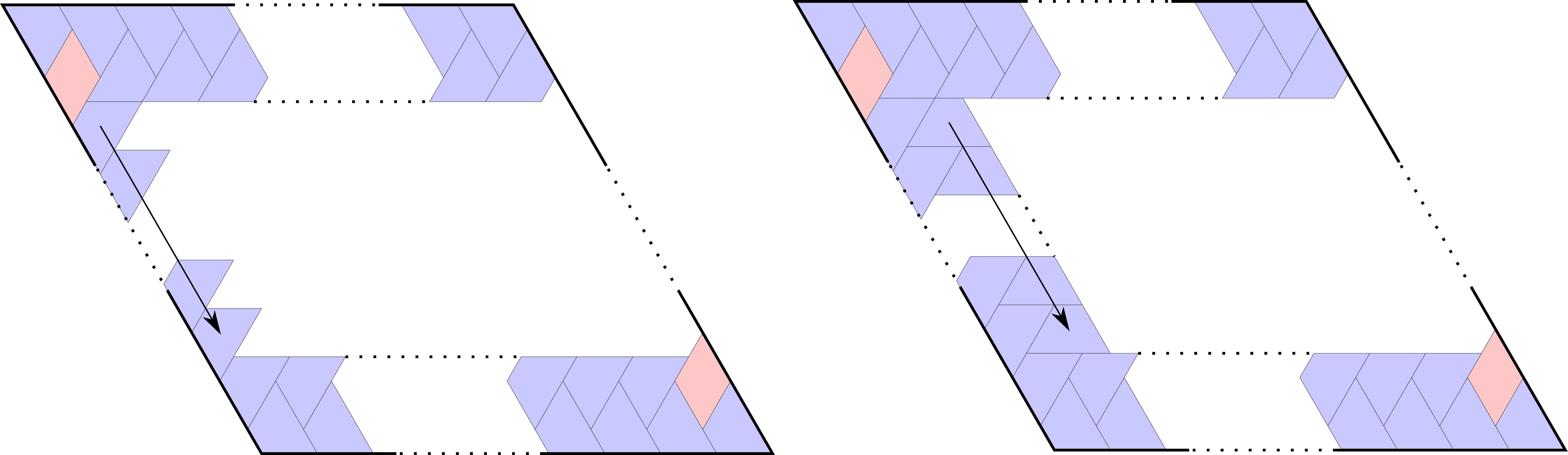}
\end{figure}  

Finally, the same is done for two rightmost columns in the reversed direction. We obtain the contour of width $2$ of the lozenge of size $N$ that is uniquely defined:

\begin{figure}[h!] 
\centering
\includegraphics[width=0.7\textwidth]{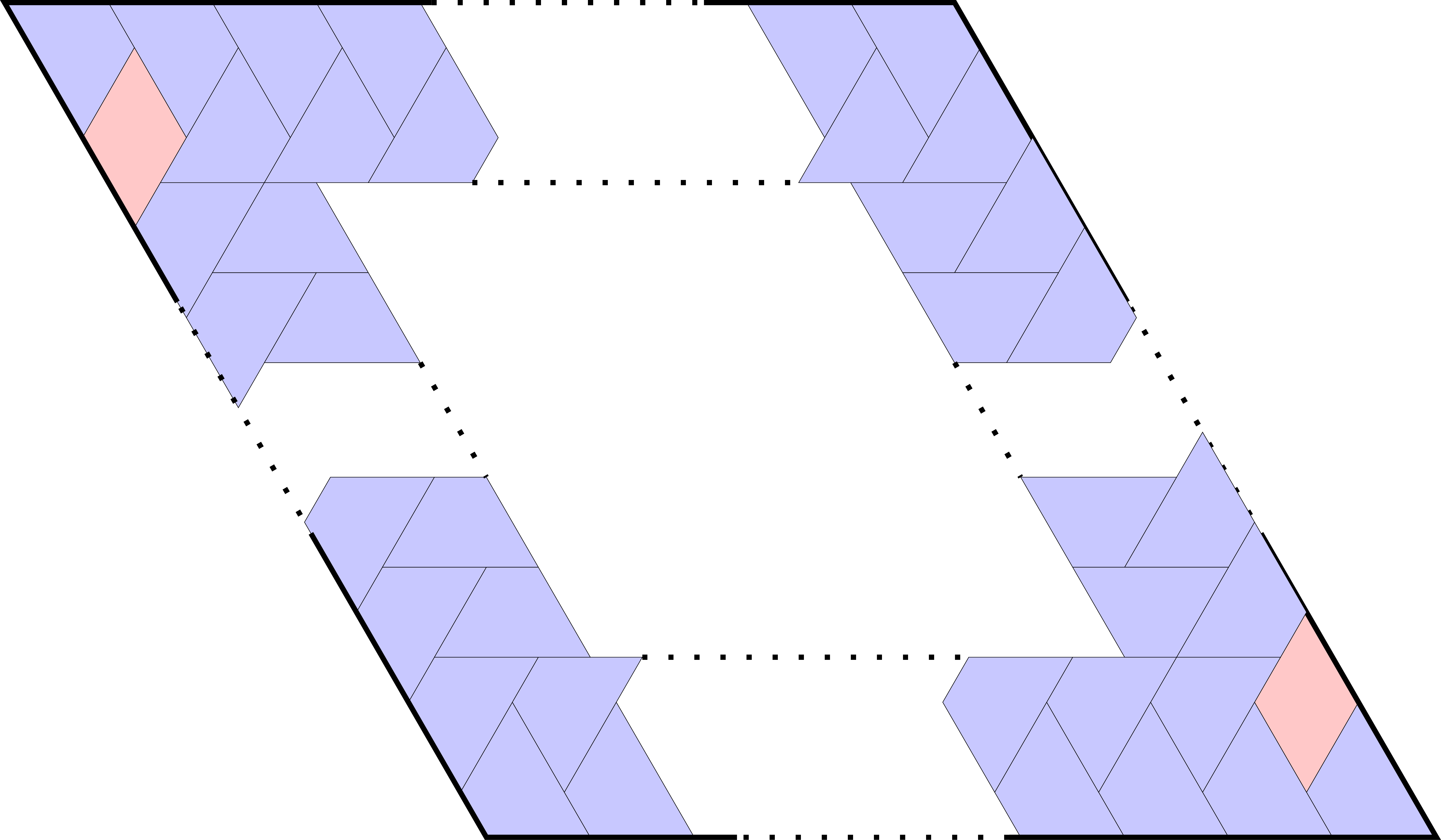}
\end{figure}  
 
The uncovered lozenge in the center has size $N-4$ and the minimal tiling for it is already uniquely defined. At the end one either gets the whole region covered, a lozenge of size $1$, $2$ or $3$.

To conclude, there is a unique minimal tiling that is accessible from every tiling of $R$ by flips, which means that $\Omega$ is connected.

\end{proof}

\textbf{Remark}. The theorem does seem to work for other regions that have shapes similar to a lozenge and flat boundary (such as trapeze, for example). But one has to change the proof of the fourth fact and show an exact way of uniquely defining the minimal tiling given the height on the boundary. It would be nice to characterize all shapes for which the tiling graph of the restrained Kagome tilings is connected and find a general proof of uniqueness of the minimal tiling.

See examples: Figure \ref{fig: restr_final} shows partial and complete minimal tiling for $N = 8$, Figure \ref{fig: k17} shows the minimal tiling for $N = 13$ and Figure \ref{fig: minimal_fish_free_lozenge} -- for $N = 50$.

\begin{figure}[hbtp] 
\centering
\includegraphics[width=1\textwidth]{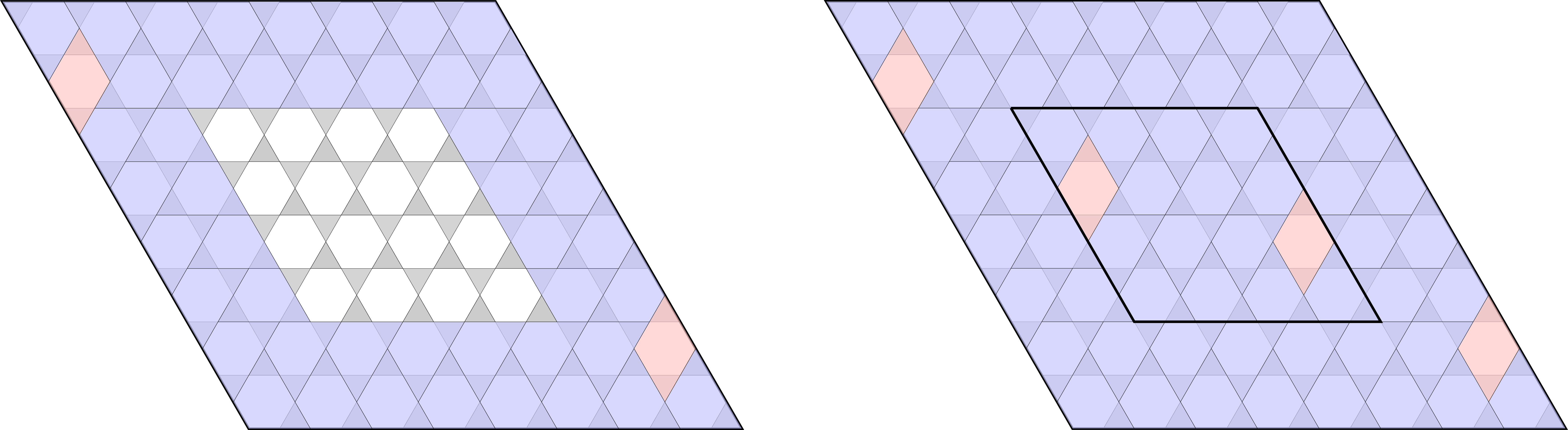}
\caption{Partial and complete minimal restrained Kagome tiling of a lozenge of size $8$.}
\label{fig: restr_final}
\end{figure}

\begin{figure}[hbtp] 
\centering
\includegraphics[width=0.7\textwidth]{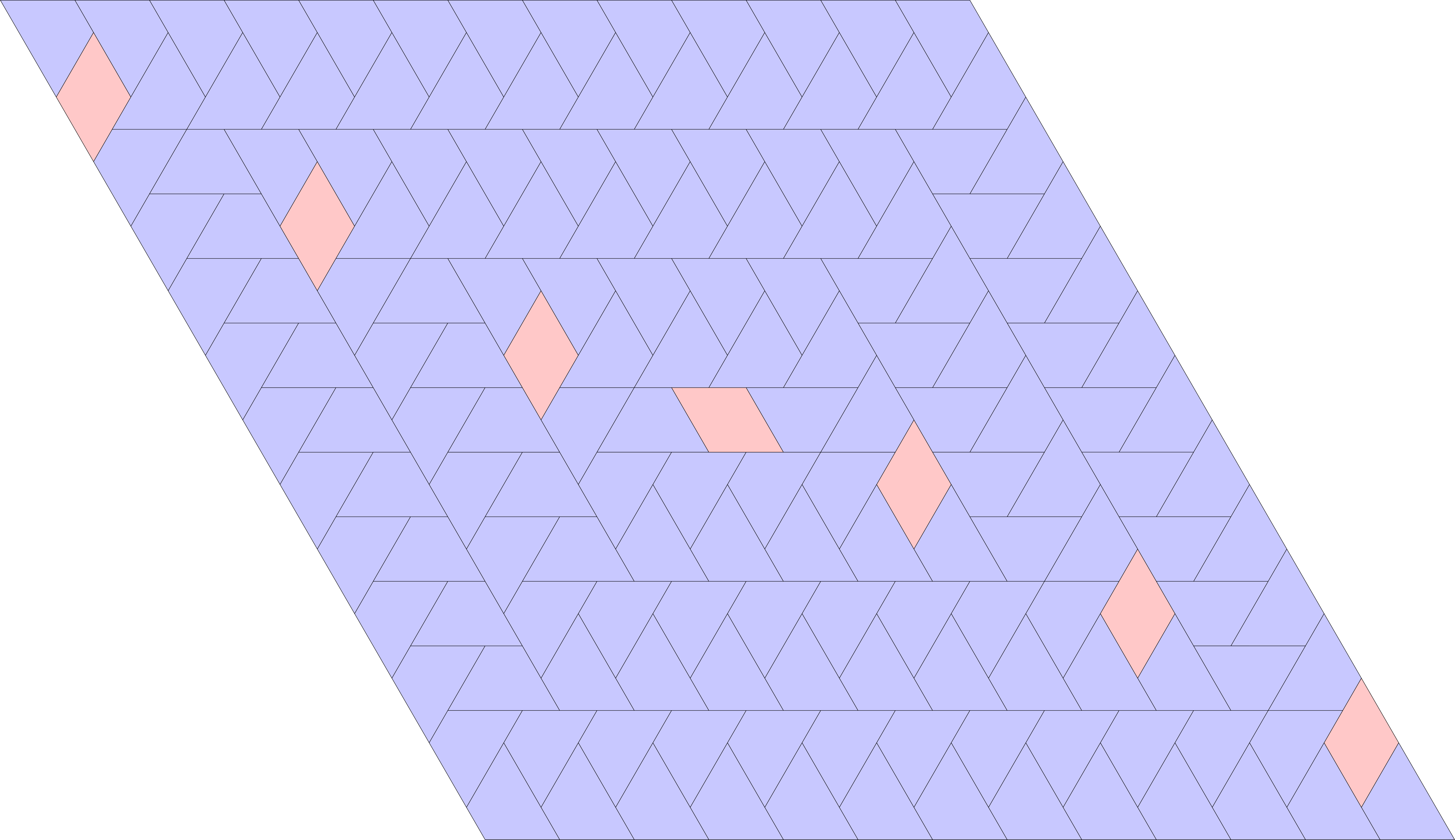}
\caption{Minimal restrained Kagome tiling of a lozenge of size $13$.}
\label{fig: k17}
\end{figure}

\begin{figure}[hbtp] 
\centering
\includegraphics[width=0.9\textwidth]{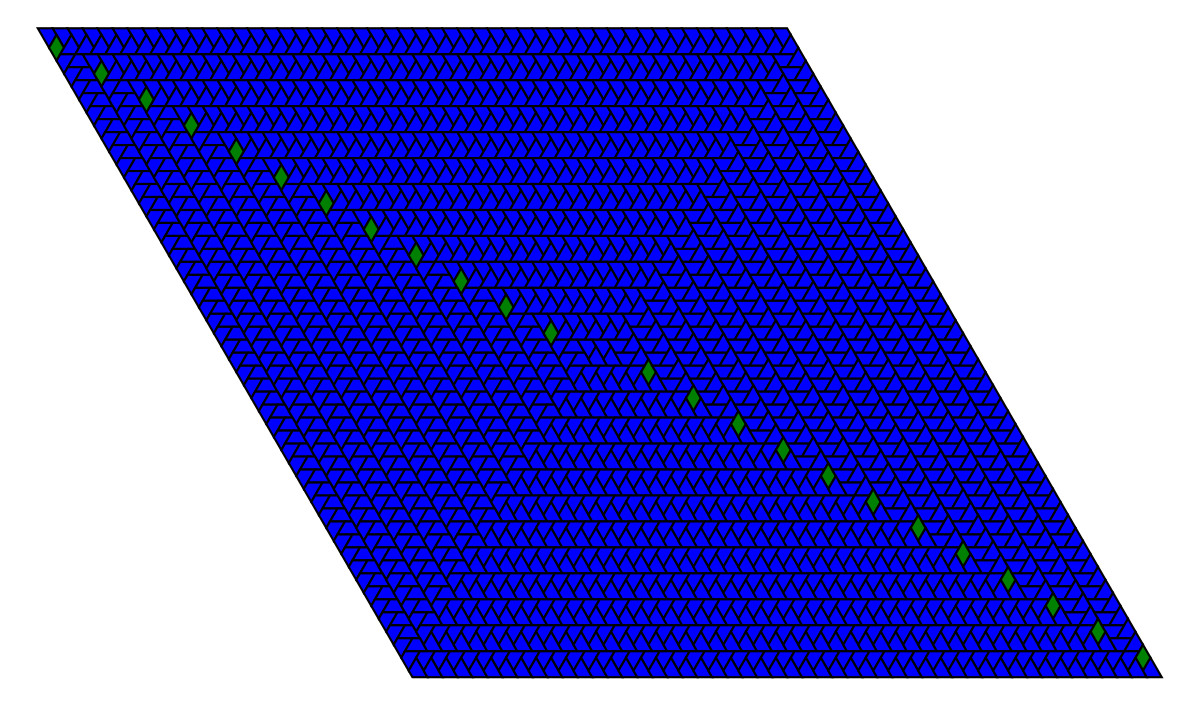}
\caption{Minimal restrained Kagome tiling of a lozenge of size $50$.}
\label{fig: minimal_fish_free_lozenge}
\end{figure}



\newpage
Figure \ref{fig: treillis} shows a tiling graph for general Kagome tilings of a small lozenge region. Restrained Kagome tilings form a connected sub-graph via black edges. Red edges mark flips that include fish tiles. The corresponding tilings complement the restrained tilings to the entire graph of all Kagome tilings.

\begin{figure}[hbtp] 
\centering
\includegraphics[width=0.6\textwidth]{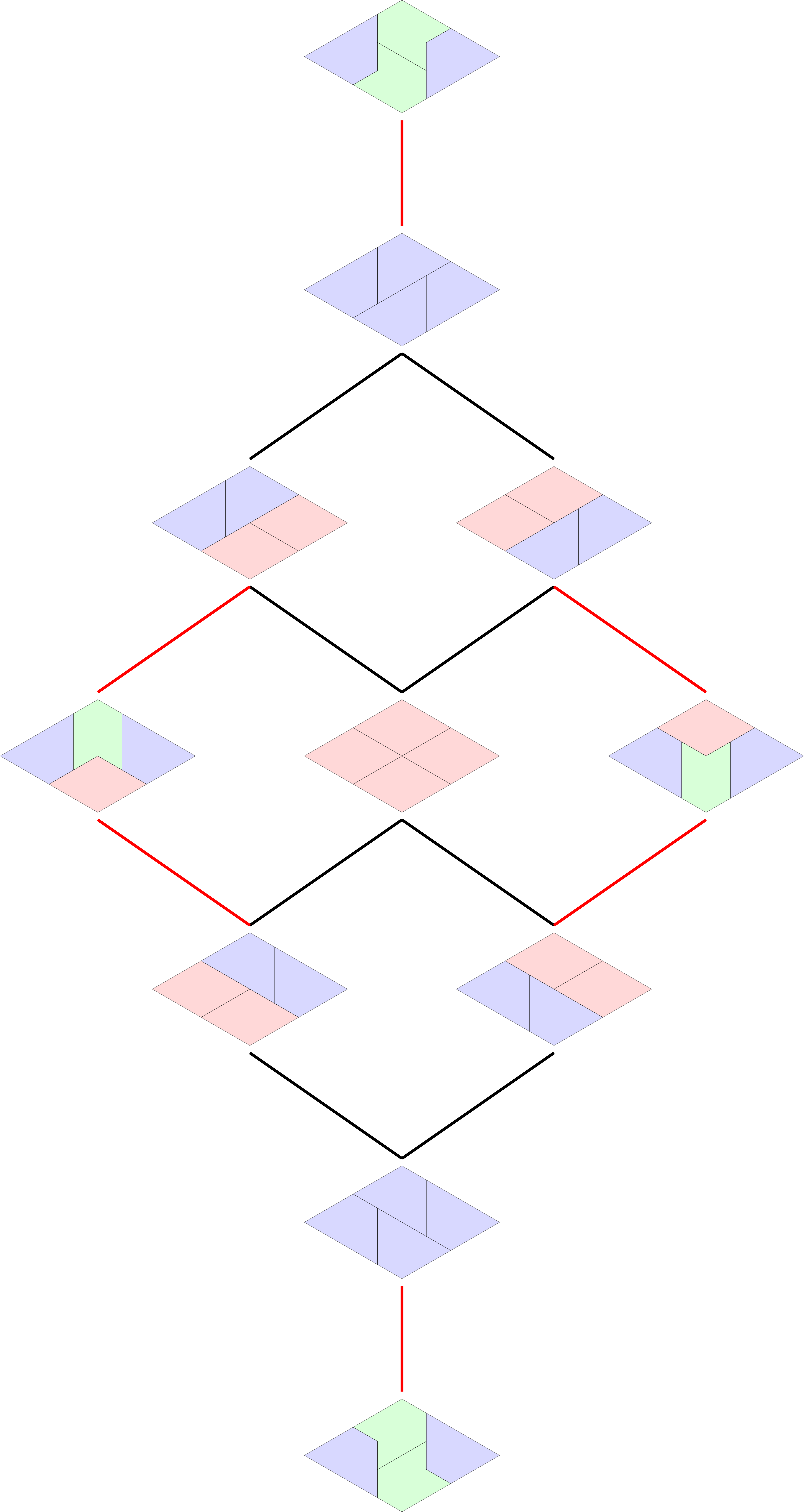}
\caption{Tiling graph for  a small lozenge region: Restrained Kagome tilings form a connected sub-graph}
\label{fig: treillis}
\end{figure}


\newpage
\begin{lemma}
MC$_{restr}$ has uniform stationary distribution.
\end{lemma}

\begin{proof}
MC$_{restr}$  is \textit{irreducible} since any tiling from $\Omega$ can be obtained from any other tiling via a finite number of consecutive flips, it is \textit{aperiodic} since the self-loop probability is greater than zero. Therefore, the chain is \textit{ergodic} and it converges to its unique stationary distribution. Moreover, the transition probabilities are symmetric, so the stationary distribution is uniform.
\end{proof}

\begin{theorem}
\label{th : fast mixing restr}
Let $R$ be a finite simply connected region of the Kagome region of area $N$ which is tileable by the restrained family of Kagome prototiles. Then MC$_{restr}$ is rapidly mixing.
More precisely, there exists a constant $c > 0$ such that the mixing time $\tau^{R}_{mix }$ of $MC_{restr}$  satisfies 
$$
\tau^{R}_{mix }(\varepsilon) \leq cN^4 \lceil \ln \varepsilon^{-1} \rceil.
$$
\end{theorem}

\begin{proof}
Consider restrained Kagome tilings and $MC_{restr}$. We prove that the chain is rapidly mixing by the use of the path coupling theorem that did not seem to work for the general case (see the end of Section \ref{sec: markov chain}).

Remember that the reason why it did not work for the general case, was that there were 4 vertices that increased $\Delta\varphi$ (as shown in  \eqref{eq: bad_coupling} and Figure \ref{fig: bad_path_coupling}). Now in this case of Figure  \ref{fig: bad_path_coupling}, for a pair of tilings $A$ and $B$ that differ by one flip around the vertice 1, no flips can be made around vertices $2,3,4,5$ which gives

\begin{equation}
\label{eq: good_coupling}
\mathbb{E}[\Delta{\varphi}] = -\frac{1}{N}.
\end{equation}

The other cases give at most one bad vertex on each side as shown in Figure \ref{fig: good_path_coupling}. This gives $ +\frac{1}{2N}$ in vertices $2$ and $4$, thus in total

\begin{equation}
\label{eq: good_coupling1}
\mathbb{E}[\Delta{\varphi}] \leq -\frac{1}{N} +\frac{1}{2N}+\frac{1}{2N} = 0.
\end{equation}
 
\begin{figure}[hbtp] 
\centering
\includegraphics[width=0.8\textwidth]{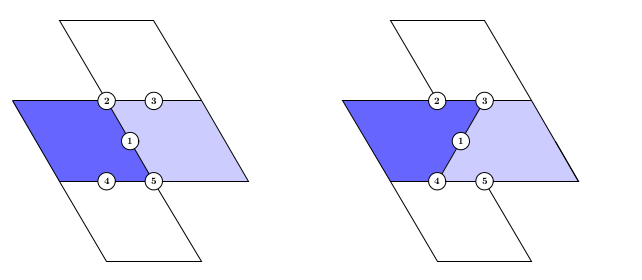}
\caption{The worst case for the path coupling for $MC_{restr}$.}
\label{fig: good_path_coupling}
\end{figure}

If $A_t \neq B_t$ for all $t$, then 
$$
\mathbb{P}[\varphi(A_{t+1}, B_{t+1}) \neq \varphi(A_{t}, B_{t})\vert A_t, B_t ] \geq \frac{1}{N},
$$
because in any case choosing the vertex $1$ will decrease the distance. The coupling lemma  gives the following bound: 
\begin{equation}
\label{eq: path_coup}
\tau^{R}_{mix}(\varepsilon) \leq \frac{D^2}{\frac{1}{N}} \lceil \ln \varepsilon^{-1} \rceil.
\end{equation}

Since for the restrained tilings, the diameter $D$ is also  $O(N^{\frac{3}{2}})$ (in the same way as for the general Kagome tiling using Proposition \ref{prop: weight for vertices}) and plugging it in \eqref{eq: path_coup} yields the desired bound:
$$
\tau^{R}_{mix}(\varepsilon) \leq cN^4 \lceil \ln \varepsilon^{-1} \rceil,
$$
where $c $ is a positive constant.

\end{proof}


\section{Limit shape}
\label{sec: limit shape}
Consider a lozenge region with a non flat boundary in such a way that the height function of the boundary is $\Omega(n)$, where $n$ is the size of the region (See Figure \ref{fig: contour}). Then there appears to be a phenomenon similar to the Arctic circle in case of dominoes and lozenges (see \cite{CLP02,JPS98}). Figures \ref{fig: actic 4}, \ref{fig: actic 5} show simulations done using the Coupling from the Past algorithm. It would be interesting to characterize the limiting shape rigorously.\\

\begin{figure}[hbtp] 
\centering
\includegraphics[width=0.4\textwidth]{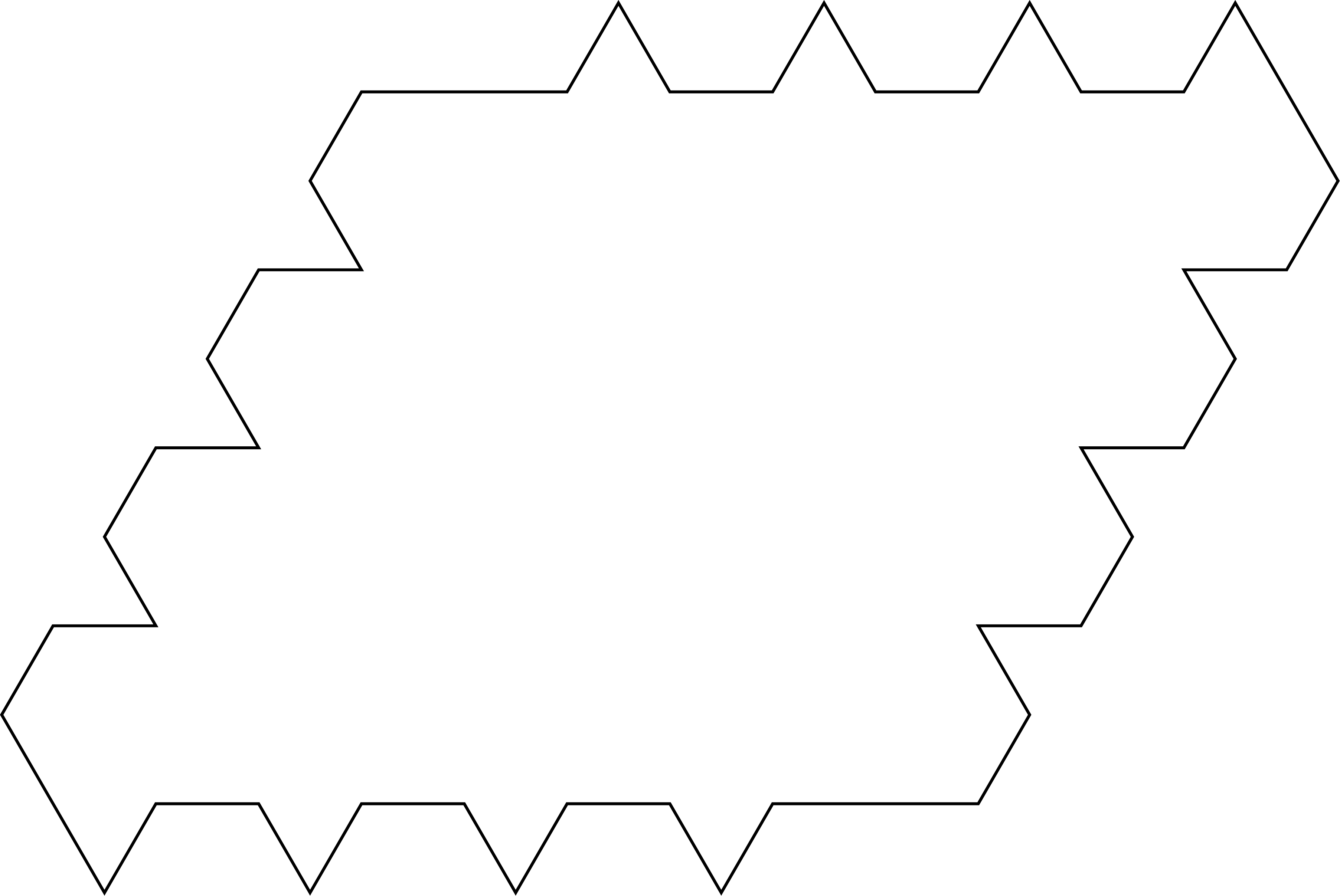}
\caption{Kagome lozenge-shaped region with a non flat boundary.}
\label{fig: contour}
\end{figure} 

\begin{figure}[hbtp] 
\includegraphics[width=1\textwidth]{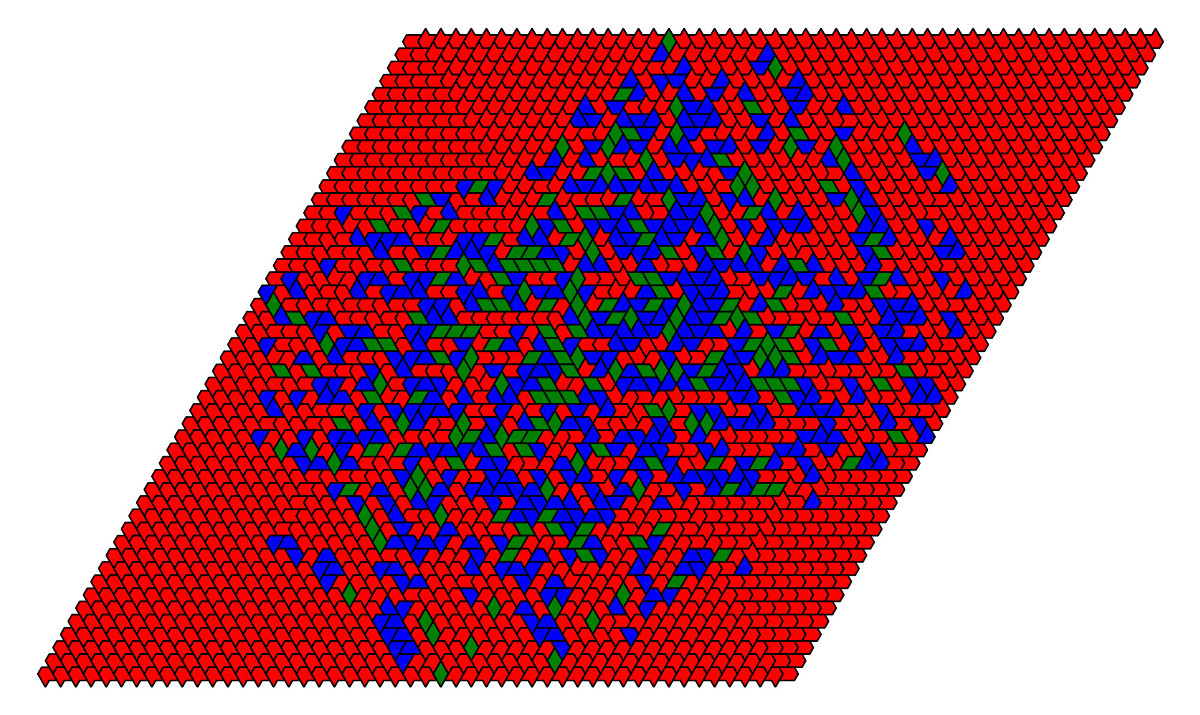}
\caption{Kagome tiling of a lozenge of size 50.}
\label{fig: actic 4}
\end{figure} 

\begin{figure}[hbtp] 
\includegraphics[width=1\textwidth]{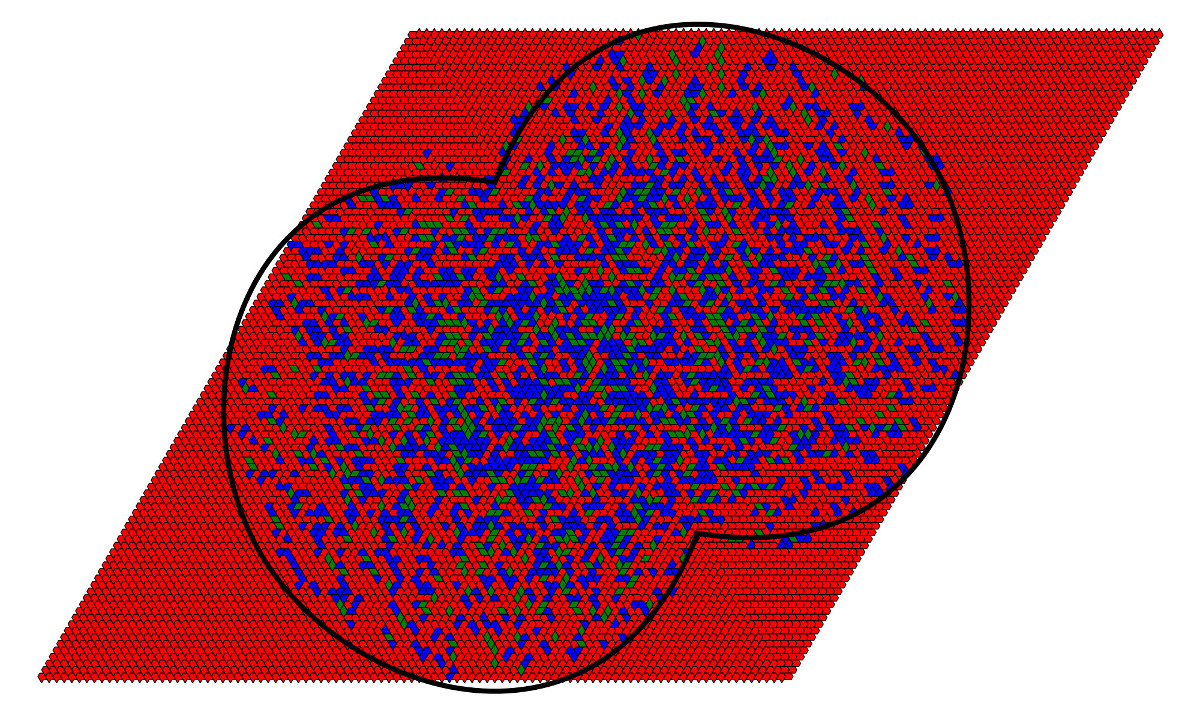}
\caption{Kagome tiling of a lozenge of size 100 with a hand-drawn ``frozen'' boundary.}
\label{fig: actic 5}
\end{figure}

\section*{Acknowledgements}

The author would to thank Thomas Fernique for valuable remarks, help with pictures and for reading the final draft of the paper, Olivier Bodini and Eric Rémila for their helpful comments.


\end{document}